\documentclass[a4paper]{article}
\usepackage{a4wide}
\usepackage[utf8]{inputenc}

\usepackage{amsmath}
\usepackage{amsfonts}
 \usepackage{amssymb}
\usepackage{mathtools}   
\usepackage{empheq}
\usepackage{dsfont}

\usepackage{stmaryrd}

\usepackage{comment}
\usepackage{enumerate}

\usepackage{myprobdescr}
\usepackage[backgroundcolor=gray!10,linecolor=gray,bordercolor=black]{todonotes}

\usepackage{longtable}
\usepackage{lscape}

\usepackage[noend]{algpseudocode}

\usepackage{amsthm}
\newtheorem{theorem}{Theorem}{\bfseries}{\normalfont}
\newtheorem{lemma}{Lemma}{\bfseries}{\normalfont}
{\bfseries}{\normalfont}
{\bfseries}{\normalfont}
\newtheorem{rrule}{Rule}{\bfseries}{\normalfont}
{\bfseries}{\normalfont}
\newtheorem{obs}{Observation}{\bfseries}{\normalfont}
\usepackage{tikz}
\usetikzlibrary{positioning,backgrounds,patterns,calc}
\usepackage{pgfplots}
\usepackage{siunitx}
\usepackage{pgfplotstable}
\usepackage{caption,subcaption}

\pgfplotsset{
    discard if not/.style 2 args={
        x filter/.code={
            \edef\tempa{\thisrow{#1}}
            \edef\tempb{#2}
            \ifx\tempa\tempb
            \else
                
            \fi
        }
    }
}

\usepackage{bbding}
\usepackage{placeins}

\usepackage{microtype,ellipsis}

 \usepackage{phonetic}
 \newcommand{\dem}{\text{\hausad}\xspace}

\usepackage{tabularx}
\usepackage{array}
\usepackage{multirow}
\newcolumntype{L}[1]{>{\raggedright\let\newline\\\arraybackslash\hspace{0pt}}m{#1}}
\newcolumntype{C}[1]{>{\centering\let\newline\\\arraybackslash\hspace{0pt}}m{#1}}
\newcolumntype{R}[1]{>{\raggedleft\let\newline\\\arraybackslash\hspace{0pt}}m{#1}}
\makeatletter
\newcommand{\thickhline}{%
    \noalign {\ifnum 0=`}\fi \hrule height 1pt
    \futurelet \reserved@a \@xhline
}
\newcolumntype{"}{@{\hskip\tabcolsep\vrule width 1pt\hskip\tabcolsep}}
\makeatother

\usepackage[pagebackref,pdfdisplaydoctitle,pagecolor=orange!40!black,menucolor=orange!40!black,filecolor=magenta!40!black,urlcolor=blue!40!black,linkcolor=red!40!black,citecolor=green!40!black,colorlinks]{hyperref}

\def\papertitle{Improved Upper and Lower Bound Heuristics for Degree Anonymization in Social Networks} 
\hypersetup{pdftitle={\papertitle}, pdfauthor={authors}}

\usepackage{algorithm}

\usepackage{units}

\usepackage[numbers,sort]{natbib}
\makeatletter
\def\NAT@spacechar{~}
\makeatother

\usepackage[sort&compress]{cleveref} 

 \newcommand{\blockSeq}{block sequence\xspace}
 \newcommand{\bigO}{\mathcal{O}}

 \newcommand{\size}{\ell}
\newcommand{\IndSet}{\textsc{Independent Set}\xspace}

\crefname{rrule}{Rule}{Rules}
\crefname{line}{Line}{Lines}
\crefname{equation}{Inequality}{Inequalities}
\crefname{section}{Section}{Sections}
\crefname{subsection}{Subsection}{Subsections}
\crefname{enumi}{Property}{Properties}

\title{\papertitle}

\author{Sepp Hartung \and Clemens Hoffmann \and André Nichterlein \\
 \\
  \multicolumn{1}{p{.9\textwidth}}{\centering{Institut f\"ur Softwaretechnik und Theoretische Informatik, TU Berlin,  Germany}}\\
\multicolumn{1}{p{.9\textwidth}}{\centering\texttt{\{sepp.hartung, andre.nichterlein, clemens.hoffmann\}@tu-berlin.de}
}}
\date{}

\usepackage{etoolbox}

\newcommand{\dash}{\nobreakdash-\hspace{0pt}}

\newcommand{\kDSA}{$k$\dash{}DSA\xspace}
\newcommand{\kDSRA}{$k$-RDSA\xspace}
\newcommand{\kDegAnon}{\textsc{Degree Anonymization}\xspace}
\newcommand{\kDegSeqRealAnon}{\textsc{Realizable $k$\dash Degree Sequence Anonymity}\xspace}
\newcommand{\kDegSeqAnon}{\textsc{$k$\dash Degree Sequence Anonymity}\xspace}
\newcommand{\kAnonymized }[1][k]{\ensuremath{#1}\dash{}anonymized\xspace}
\newcommand{\kAnonymous}[1][k]{\ensuremath{#1}\dash{}anonymous\xspace}
\newcommand{\kAnonymization}{$k$\dash{}anonymization\xspace}
\newcommand{\RealProb}{\textsc{Degree Realization}\xspace}

\newcommand{\TS}{\tilde{S}}
\newcommand{\TG}{\tilde{G}}

\newcommand{\degreeBound}{2\Delta^2}

\newcommand{\kInsertSet}[1][k]{\ensuremath{#1}\dash insertion set\xspace}

\newcommand{\cost}{\operatorname{cost}}

\newcommand{\D}{\mathcal{D}}
\newcommand{\B}{\mathcal{B}}

\newcommand{\N}{\mathds{N}}

\newcommand{\decprob}[3]{%
	\begin{center}
		\begin{minipage}{\textwidth}
		\defDecprob{#1}{#2}{#3}
		\end{minipage}
	\end{center}
}
\newcommand{\optprob}[3]{%
	\begin{center}
		\begin{minipage}{\textwidth}
		\defOptprob{#1}{#2}{#3}
		\end{minipage}
	\end{center}
}

\newcommand{\emptyBox}{}

\usepackage{caption}

\begin{document}

\maketitle

\setcounter{footnote}{0}

\begin{abstract}
	Motivated by a strongly growing interest in anonymizing social network data, we investigate the NP-hard \kDegAnon problem: given an undirected graph, the task is to add a minimum number of edges such that the graph becomes \emph{\kAnonymous.}
	That is, for each vertex there have to be at least~$k-1$ other vertices of exactly the same degree.
%
	The model of degree anonymization has been introduced by Liu and Terzi~[ACM SIGMOD'08], who also proposed and evaluated a two-phase heuristic.
	We present an enhancement of this heuristic, including new algorithms for each phase which significantly improve on the previously known theoretical and practical running times.
	Moreover, our algorithms are optimized for large-scale social networks and provide upper and lower bounds for the optimal solution.
	Notably, on about~26\,\%~of the real-world data we provide (provably) optimal solutions; whereas in the other cases our upper bounds significantly improve on known heuristic~solutions.
\end{abstract}

\section{Introduction}

In recent years, the analysis of (large-scale) social networks received a steadily growing attention and turned into a very active research field~\cite{EK10}.
Its importance is mainly due the
easy availability of social networks and due to the potential gains of an analysis revealing important subnetworks, statistical information, etc.
However, as the analysis of networks may reveal sensitive data about the involved users, before publishing the networks it is necessary to preprocess them in order to respect privacy issues~\cite{FWCY10}.
In a landmark paper~\cite{LT08} initiating a lot of follow-up work~\cite{CGSV12,LSB12,HNNS13},\footnote{According to Google Scholar (accessed Feb.~2014) it has been cited more than 300 times.} \citeauthor{LT08} transferred the so-called \emph{k-anonymity} concept known for tabular data in databases~\cite{SS98,Sam01,Swe02,FWCY10} to social networks modeled as undirected graphs.
A graph is called \emph{\kAnonymous} if for each vertex there are at least~$k-1$ other vertices of the same degree.
Therein, the larger~$k$ is, the better the expected level of anonymity is.

In this work we describe and evaluate a combination of heuristic algorithms which provide (for many tested instances matching) lower and upper bounds, for the following NP-hard graph anonymization problem:

\optprob{\kDegAnon \cite{LT08}}
	{An undirected graph~$G=(V,E)$ and an integer~$k\in \N$.}
	{Find a minimum-size edge set~$E'$ over~$V$ such that adding~$E'$ to~$G$ results in a \kAnonymous graph.}
As \kDegAnon is NP-hard even for constant~$k\ge 2$~\cite{HNNS13}, all known (experimentally evaluated) algorithms, are heuristics in nature~\cite{LT08,LSB12,RHT13,TY09}.
\citet{LT08} proposed a heuristic which, in a nutshell, consists of the following two phases: i)~Ignore the graph structure and solve a corresponding number problem and ii)~try to transfer the solution from the number problem back to the graph instance.
\begin{figure}[t]
	\begin{center}
		\begin{tabular}{ccccccc}
			\begin{tikzpicture}[scale=1]
				\tikzstyle{knoten}=[circle,draw,fill=black!20,minimum size=5pt,inner sep=2pt]

				\node[knoten] (K-1) at (0,0) {};
				\node[knoten] (K-2) at (0,1) {};
				\node[knoten] (K-3) at (1,0) {};
				\node[knoten] (K-4) at (1,1) {};
				\foreach \i / \j in {1/2, 2/3, 2/4, 3/4}{
					\path (K-\i) edge[-] (K-\j);
				}
			\end{tikzpicture}
			&
			\begin{tikzpicture}[scale=1]
				\tikzstyle{knoten}=[minimum size=5pt,inner sep=2pt]
				\node[knoten] (K-1) at (0,0) {};
				\node[knoten] (K-2) at (0,1) {};
				\node at (0,0.5) {$\overset{\phantom{k}}{\Rightarrow}$};
			\end{tikzpicture}
			&
			\begin{tikzpicture}[scale=1]
				\tikzstyle{knoten}=[minimum size=5pt,inner sep=2pt]
				\node[knoten] (K-1) at (0,0) {};
				\node[knoten] (K-2) at (0,1) {};
				\node at (0,0.35) {1,2,2,3};
			\end{tikzpicture}
			&
			\begin{tikzpicture}[scale=1]
				\tikzstyle{knoten}=[minimum size=5pt,inner sep=2pt]
				\node[knoten] (K-1) at (0,0) {};
				\node[knoten] (K-2) at (0,1) {};
				\node at (0,0.5) {$\overset{Phase~1.}{\Rightarrow}$};
			\end{tikzpicture}
			&
			\begin{tikzpicture}[scale=1]
				\tikzstyle{knoten}=[minimum size=5pt,inner sep=2pt]
				\node[knoten] (K-1) at (0,0) {};
				\node[knoten] (K-2) at (0,1) {};
				\node at (0,0.35) {3,3,3,3};
			\end{tikzpicture}
			&
			\begin{tikzpicture}[scale=1]
				\tikzstyle{knoten}=[minimum size=5pt,inner sep=2pt]
				\node[knoten] (K-1) at (0,0) {};
				\node[knoten] (K-2) at (0,1) {};
				\node at (0,0.5) {$\overset{Phase~2.}{\Rightarrow}$};
			\end{tikzpicture}
			&
			\begin{tikzpicture}[scale=1]
				\tikzstyle{knoten}=[circle,draw,fill=black!20,minimum size=5pt,inner sep=2pt]

				\node[knoten] (K-1) at (0,0) {};
				\node[knoten] (K-2) at (0,1) {};
				\node[knoten] (K-3) at (1,0) {};
				\node[knoten] (K-4) at (1,1) {};

				\foreach \i / \j in {1/2, 2/3, 2/4, 3/4, 1/3, 1/4}{
					\path (K-\i) edge[-] (K-\j);
				}
			\end{tikzpicture}
			\\
			input graph~$G$ & & degree   & & ``\kAnonymized''  & & realization \\
			with $k=4$  & & sequence~$\D$ & & degree  sequence~$\D'$      & & of~$\D'$ in~$G$\\
		\end{tabular}
	\end{center}\vspace{-10pt}
	\caption{A simple example for the two phases in the heuristic of \citet{LT08}. Phase~1:~Anonymize the degree sequence~$\D$ of the input graph~$G$ by increasing the numbers in it such that each resulting number occurs at least~$k$ times. Phase~2: Realize the \kAnonymized degree sequence~$\D'$ as a super-graph of~$G$.}
	\label{fig:heuristicExample}
\end{figure}
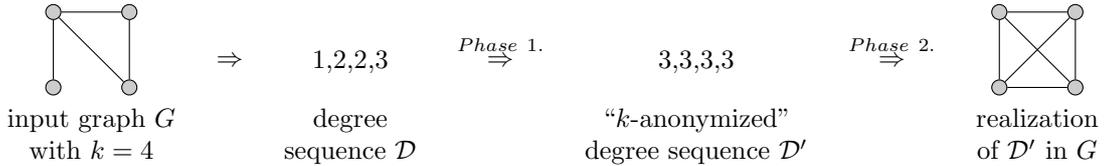
More formally (see \autoref{fig:heuristicExample} for an example), given an instance~$(G,k)$, first compute the \emph{degree sequence}~$\D$ of~$G$, that is, the multiset of positive integers corresponding to the vertex degrees in~$G$.
Then, Phase~1 consists of $k$\dash{}anonymizing the degree sequence~$\D$ (each number occurs at least~$k$ times) by a minimum amount of increments to the numbers in~$\D$ resulting in~$\D'$.
In Phase~2, try to realize the $k$\dash{}anonymous sequence~$\D'$ as a super-graph of~$G$, meaning that each vertex gets a \emph{demand}, which is the difference of its degree in~$\D'$ compared to~$\D$, and then a ``realization'' algorithm adds edges to~$G$ such that for each vertex the amount of incident new edges equals its demand.

Note that, since the minimum ``\kAnonymization cost'' of the degree sequence~$\D$ (sum over all demands) is always a lower bound on the \kAnonymization cost of~$G$, the above described algorithm, if successful when trying to realize~$\D'$ in~$G$, optimally solves the given \mbox{\kDegAnon~instance}.
%

\smallskip\noindent{\bf Related Work.}
We only discuss work on \kDegAnon directly related to what we present here.
Our algorithm framework is based on the two-phase algorithm due to~\citet{LT08} where also the model of graph \mbox{(degree-)anonymization} has been introduced.
Other models of graph anonymization have been studied as well, see \citet{ZP11} (studying the neighborhood of vertices) and \citet{CGSV12} (anonymizing vertex subsets).
We refer to \citet{ZPL08} for a survey on anonymization techniques for social networks.
\kDegAnon is NP-hard for constant~$k\ge 2$ and it is W[1]-hard (presumably not fixed-parameter tractable) with respect to the parameter size of a solution size~\cite{HNNS13}. 
On the positive side, there is polynomial-size kernel (efficient and effective preprocessing) with respect to the maximum degree of the input graph~\cite{HNNS13}. 
\citet{LSB12} and \citet{RHT13} designed and evaluated heuristic algorithms that are
our reference points for comparing our results.

\smallskip\noindent{\bf Our Contributions.}
Based on the two-phase approach of \citet{LT08} we significantly improve the lower bound provided in Phase~1 and provide a simple heuristic for new upper bounds in Phase~2.
Our algorithms are designed to deal with large-scale real world social networks (up to half a million vertices) and exploit some common features of social networks such as the power\dash{}law degree distribution~\cite{BA99}.
%
For Phase~1, we provide a new dynamic programming algorithm of $k$\dash{}anonymizing a degree sequence~$\D$ ``improving'' the previous running time~$\bigO(nk)$ to $\bigO(\Delta k^2 s)$.
Note that maximum degree~$\Delta$ is in our considered instances about 500 times smaller than the number of vertices~$n$.
We also implemented a data reduction rule which
leads to significant speedups of the dynamic program.
We study two different cases to obtain upper bounds.
If one of the degree sequences computed in Phase~1 is realizable, then this gives an optimal upper bound and otherwise
we heuristically look for ``near'' realizable degree sequences.
For Phase~2 we evaluate the already known ``local exchange'' heuristic~\cite{LT08} and provide some theoretical justification of its quality.

We implemented our algorithms and compare our upper bounds with a heuristic of \citet{LSB12}, called \emph{clustering-heuristic} in the following.
Our empirical evaluation demonstrates that in about 26\% of the real-world instances the lower bound matches the upper bound and in the remaining instances our heuristic upper bound is on average 40\% smaller than the one provided by the clustering-heuristic.
However, this comes at a cost of increased running time: the clustering-heuristic could solve all instances within~15~seconds whereas there are a few instances where our algorithms could not compute an upper~bound~within~one~hour.

Due to the space constraints, all proofs and some details are deferred to an appendix.

\section{Preliminaries}\label{sec:prelim}

We use standard graph-theoretic notation.
All graphs studied in this paper are undirected and simple without self-loops and multi-edges.
For a given graph~$G=(V,E)$ with vertex set~$V$ and edge set~$E$ we set~$n:=|V|$ and~$m:=|E|$.
Furthermore, by~$\deg_G(v)$ we denote the degree of a vertex~$v\in V$ in~$G$ and $\Delta_G$ denotes the maximum degree in~$G$.
For $0\le d\le \Delta_G$ let $B^G_d := \{v \in V \mid \deg_G(v)=d\}$ be the \emph{block} of degree~$d$, that is, the set of all vertices with degree~$d$ in~$G$.
Thus, being \kAnonymous is equivalent to each block being of size either zero or at least~$k$.
For a set~$S$ of edges with endpoints in a graph~$G$, we denote by~$G+S$ the graph that results from inserting all edges from~$S$ into~$G$.
We call~$S$ an \emph{edge insertion set} for~$G$, and if~$G+S$ is \kAnonymous, then it is an \emph{\kInsertSet}.

A \emph{degree sequence}~$\D$ is a multiset of positive integers and $\Delta_\D$ denotes its maximum value.
The degree sequence of a graph~$G$ with vertex set~$V = \{v_1, \ldots, v_n\}$ is~$\D_G := \{\deg_G(v_1), \ldots, \deg_G(v_n)\}$.
For a degree sequence~$\D$, we denote by~$b_d$ how often value~$d$ occurs in~$\D$ and we set $\B=\{b_0,\ldots,b_{\Delta_\D}\}$ to be the \emph{\blockSeq} of~$\D$, that is, $\B$ is just the list of the block sizes of~$G$.
Clearly, the block sequence of a graph~$G$ is the block sequence of $G$'s degree sequence.
The \blockSeq can be viewed as a compact representation of a degree sequence (just storing the amount of vertices for each degree) and we use these two representations of vertex degrees interchangeably.
Equivalently to graphs, a block sequence is \kAnonymous if each value is either zero or at least~$k$ and a degree sequence is \kAnonymous if its corresponding block sequence is \kAnonymous.

Let $\D=\{ d_1, \ldots, d_n\}$ and $\D' =\{ d'_1, \ldots, d'_n\}$ be two degree sequences with corresponding block sequences~$\B$ and~$\B'$.
We define $\|\B\| = |\D| = \sum_{i=1}^n d_i$.
We write $\D'\ge \D$ and $\B'\ogreaterthan \B$ if for both degree sequences---sorted in ascending order---it holds that $d'_i\ge d_i$ for all~$i$.
Intuitively, this captures the interpretation ``$\D'$ can be obtained from~$\D$ by increasing some values''.
If $\D'\ge \D$, then (for sorted degree sequences) we define the degree sequence $\D'-\D =\{ d'_1-d_1, \ldots, d'_n-d_n\}$ and set $\B' \circleddash \B$ to be its block sequence.
We omit sub- and superscripts if the graph is clear from the context.

%
%
%

\section{Description of the Algorithm Framework}\label{sec:alg-desc}

In this section we present the details of our algorithm framework to solve \kDegAnon.
We first provide a general description how the problem is split into several subproblems (basically corresponding to the two-phase approach of \citet{LT08}) and then describe the corresponding algorithms in detail.

\subsection{General Framework Description}


We first provide a more formal description of the two-phase approach due to \citet{LT08} and then describe how we refine it:
Let $(G=(V,E),k)$ be an input instance of \kDegAnon.
\begin{description}
 \item[Phase~1:] 
  For the degree sequence~$\D$ of~$G$, compute a \kAnonymous degree sequence~$\D'$ such that $\D'\ge \D$ and $|\D-\D'|$ is minimized.
 \item[Phase~2:] Try to realize~$\D'$ in~$G$, that is, try to find an edge insertion set~$S$ such that the degree sequence of~$G+S$ is~$\D'$.
\end{description}
The minimum \kAnonymization cost of~$\D$, formally $|\D'-\D|/2$, is a lower bound on the number of edges in a  \kInsertSet for~$G$.
Hence, if succeeding in Phase~2 to realize~$\D'$, then a minimum-size \kInsertSet~$S$ for~$G$ has been found.

\citet{LT08} gave a dynamic programming algorithm which exactly solves Phase~1 and they provided the so-called local exchange heuristic algorithm for~Phase~2.
If Phase~2 fails, then the heuristic of \citet{LT08} relaxes the constraints and tries to find a \kInsertSet yielding a graph ``close'' to~$\D'$.

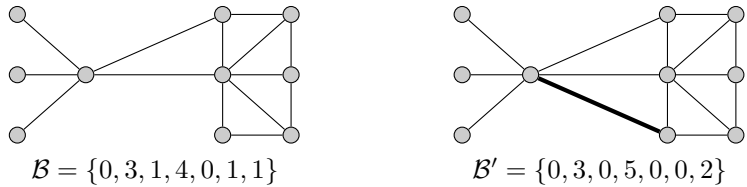
\begin{figure}[t]
	\centering
		\tikzstyle{vertex}=[circle,draw=black,minimum size=6pt,inner sep=0.8pt,fill=black!20]
		\begin{tikzpicture}[xscale=0.9,yscale=0.8]
			\foreach[count=\i] \x/\y in {0/0,0/1,0/2,1/1,3/1,3/0,3/2,4/2,4/1,4/0}
				\node[vertex](v\i) at (\x,\y) {};
			\node at (2,-0.6) {$\B =\{ 0,3,1,4,0,1,1\}$};
			\foreach \from/\to in {1/4,2/4,3/4,4/5,4/7,5/6,5/7,5/8,5/9,5/10,6/10,10/9,9/8,7/8}
				\draw (v\from) -- (v\to);
			\begin{scope}[xshift=6.5cm]
				\foreach[count=\i] \x/\y in {0/0,0/1,0/2,1/1,3/1,3/0,3/2,4/2,4/1,4/0}
					\node[vertex](u\i) at (\x,\y) {};
				\node at (2,-0.6) {$\B' =\{ 0,3,0,5,0,0,2\}$};
				\foreach \from/\to in {1/4,2/4,3/4,4/5,4/7,5/6,5/7,5/8,5/9,5/10,6/10,10/9,9/8,7/8}
					\draw (u\from) -- (u\to);
			\end{scope}
			\draw[ultra thick] (u4) -- (u6);
		\end{tikzpicture}\vspace{-5pt}
	\caption{A graph (left side) with block sequence~$\B$ that can be 2-anonymized by adding one edge (right side) resulting in~$\B'$. Another 2-anonymous \blockSeq (also of cost~two) that will be found by the dynamic programming is~$\B'' =\{ 0,2,2,4,0,0,2\}$.
	The realization of~$\B''$ in~$G$ would require to add an edge between a degree-five vertex (there is only one) and a degree-one vertex, which is impossible.}\label{fig:exampleMultipleMinSolutionsNotAllRealizable}%
\end{figure}%
We started with a straightforward implementation of the dynamic programming algorithm and the local exchange heuristic.
We encountered the problem that, even when iterating through all minimum \kAnonymous degree sequences~$\D'$, one often fails to realize~$\D'$ in Phase~2. 
More importantly, we observed the difficulty that iterating through all minimum sequences is often to time consuming because the same sequence is recomputed multiple times.
This is because the dynamic program iterates through all possibilities to choose ``sections'' of consecutive degrees in the (sorted) degree sequence~$\D$ that end up in the same block in~$\D'$.
These sections have to be of length at least~$k$ (the final block has to be full) but at most~$2k-1$ (longer sections can be split into two).
However, if there is a huge block~$B$ (of size~$\gg2k$) in~$\D$, then the algorithm goes through all possibilities to split~$B$ into sections, although it is not hard to show that at most~$k-1$ degrees from each block are increased.
 Thus, 
different ways to cut these degrees into sections result in the same degree sequence.

We thus redesigned the dynamic program for Phase~1. 
The main idea is to consider the block sequence of the input graph and exploiting the observation that at most~$k-1$ degrees from a block are increased in a minimum-size solution.
Therefore, we avoid to partition one block into multiple sections and the running time dependence on the number of vertices~$n$ can be replaced by the maximum degree~$\Delta$, yielding a significant performance increase.

We also improved the lower bound provided by $\D'-\D$ on the \kAnonymization cost of~$G$.
To this end, the basic observation was that while trying to realize one of the minimum \kAnonymous sequences~$\D'$ in~Phase~2 (failing in almost all cases), we encountered that by a simple criterion on the sequence $\D'-\D$ one can even prove that~$\D'$ is not realizable in~$G$.
That is, a \kInsertSet~$S$ for~$G$ corresponding to~$\D'$ would induce a graph with degree sequence~$\D'-\D$.
Hence, the requirement that there is a graph with degree sequence~$\D'-\D$ is a necessary condition to realize~$\D'$ in~$G$ in~Phase~2.
Thus, for increasing cost~$c$, by iterating through all \kAnonymous sequences~$\D'$ with~$|\D'-\D|=c$ and excluding the possibility that~$\D'$ is realizable in~$G$ by the criterion on~$\D'-\D$, one can step-wisely improve the lower bound on the \kAnonymization cost of~$G$.
We apply this strategy and thus our dynamic programming table allows to iterate through all \kAnonymous sequences~$\D'$ with $|\D'-\D|=c$.
Unfortunately, even this criterion might not be sufficient because the already present edges in~$G$ might prevent the insertion of a \kInsertSet which corresponds to~$\D'-\D$ (see \autoref{fig:exampleMultipleMinSolutionsNotAllRealizable} for an example).
We thus designed a test which not only checks whether~$\D'-\D$ is realizable but also takes already present edges in~$G$ into account while preserving that $|\D'-\D|$ is a lower bound on the \kAnonymization cost of~$G$.
With this further requirement on the resulting sequences~$\D'$ of Phase~1, in our experiments we observe that Phase~2 of realizing~$\D'$ in~$G$ is in 26\,\% of the real-world instances successful.
Hence, 26\,\% of the instances can be solved~optimally.
See \Cref{ssect:step1} for a detailed description of our algorithm for Phase~1.

For Phase~2 the task is to decide whether a given \kAnonymization~$\D'$ can be realized in~$G$. 
As we will show that this problem is NP-hard, we split the problem into two parts and try to solve each part separately by a heuristic.
First, we find a degree-vertex mapping, that is, we assign each degree $d'_i\in \D'$ to a vertex~$v$ in~$G$ such that $d'_i\ge \deg_G(v)$. Then, the demand of vertex~$v$ is set to~$d'_i-\deg_G(v)$.
Second, given a degree-vertex mapping with the corresponding demands
we try the find an edge insertion set such that the number of incident new edges for each vertex is equal to its demand.
While the second part could in principle be done optimally in polynomial-time by solving an $f$\dash{}factor problem~\cite{HNNS13}, we show that already a heuristic refinement of the ``local exchange'' heuristic due to~\citet{LT08} is able to succeed in most cases. 
Thus, theoretically and also in our experiments, the ``hard part'' is to find a good degree-vertex mapping.
Roughly speaking, the difficulties are that, according to~$\D'$, there is more than one possibility of how many vertices from degree~$i$ are increased to degree~$j>i$.
Even having settled this it is not clear which vertices to choose from block~$i$.
See \Cref{ssec:phase-2-des} for a detailed description of our algorithm for Phase~2.

\subsection{Phase~1: Exact $k$-Anonymization of Degree Sequences}\label{ssect:step1}
We start with providing a formal problem description of $k$\dash{}anonymizing a degree sequence~$\D$ and describe our dynamic programming algorithm to find such sequences~$\D'$.
We then describe the criteria that we implemented to improve the lower bound $|\D'-\D|$. 
%

\smallskip\noindent{\bf Basic Number Problem.}
The decision version of the degree sequence anonymization problem reads as follows.
\decprob{\kDegSeqAnon (\kDSA)}
	{A \blockSeq $\B$ and integers~$k,s\in \N$.}
	{Is there a \kAnonymous \blockSeq $\B'\ogreaterthan\B$ such that \mbox{$\|\B'\circleddash\B\|=s$?}}
The requirements on~$\B'$ in the above definition ensure that~$\B'$ can be obtained by performing exactly~$s$ many increases to the degrees in~$\B$.
\citet{LT08} give a dynamic programming algorithm that solves \kDSA optimally in~$\bigO(nk)$ time and space.
Here, besides using block instead of degree sequences, we added another dimension to the dynamic programming table storing the cost of a solution.

\begin{lemma} \label{lem:dynProgForkDSA}
\kDegSeqAnon can be solved in $\bigO(\Delta \cdot k^2 \cdot s)$ time and $\bigO(\Delta \cdot k \cdot s)$ space.
\end{lemma}
{
\begin{proof}
Let $(\B,k,s)$ be an instance of \kDegSeqAnon.
We describe a dynamic programming algorithm.
The algorithm maintains a table~$T$ where the entry $T[i,t,c]$ with $0\le i\le \Delta$, $0\le c\le s$, and $0\le t<2k$ is true
if and only if the \blockSeq $\B(i) =\{B_0,\ldots,B_i\}$ minus the last~$t$ degrees can be \kAnonymized with cost exactly~$c$.
Formally, $\B(i)$ minus the last~$t$ degrees is the block sequence~$\B'(i)$ corresponding to the degree sequence~$\D$ that is obtained from~$\D_{\B(i)}$ by removing the~$t$ highest degrees.
For we $0\le i < \Delta$ we denote by $\cost(i,t)$ the cost to increase the last~$t$ degrees in $\B(i)$ to~$i+1$.
We compute~$T[i,t,c]$ with the following recursion.
\[
 T[i,t,c]=
\begin{cases}
	c = 0 \wedge (|B_0|-t = 0 \vee |B_0|-t \ge k), & i=0 \\
	\!\begin{aligned}
		\exists t' \in \N \colon & k-(|B_i|-t)\le t'<2k \; \wedge \\
		& T[i-1,t',c-\cost(i-1,t')]=\text{true,} %
	\end{aligned} & |B_i|>t \\
  T[i-1,t-|B_i|,c],& |B_i|\le t.
\end{cases}
\]
The entry~$T[\Delta,0,s]$ is true if and only if~$(\B,k,s)$ is a yes-instance.
We compute the $\cost$-function in a preprocessing step in~$\bigO(\Delta \cdot k^2)$ time and~$\bigO(\Delta \cdot k)$ space.
Having computed the~$\cost$-function each table entry in~$T$ can be computed in~$\bigO(k)$ time.
As there are~$\Delta \cdot k \cdot s$ table entries, the overall running time is~$\bigO(\Delta \cdot k^2 \cdot s)$.

As to the correctness, observe that if~$i = 0$ in the recursion, then $c$ has to be positive and exactly~$t$ degrees of block~$B_0$ have to be increased by at least one, which is possible only if $|B_0|-t = 0$ or $|B_0|-t \ge k$.
If~$i > 0$ and~$|B_i| \le t$, then clearly all degrees in the block~$B_i$ have to be increased.
If~$i > 0$ and~$|B_i| > t$, then there remain some degrees in the block~$B_i$.
Thus, the block has to be of size at least~$k$ implying that at least~$\min(0, k - (|B_i|-t))$ degrees have to be added to~$B_i$.
Furthermore, adding more than~$2k$ vertices is not necessary, as in this case we would add only~$k$ vertices to ensure that both the blocks~$B_i$ and~$B_{i-1}$ are large enough, that is, $|B'_i| \ge k$ and~$|B'_{i-1}|\ge k$.
The correctness now follows from the fact that the recursion tries all possibilities for the value~$t'$ between~$\min(0, k - (|B_i|-t))$ and~$2k$.\qed
\end{proof}}
There might be multiple minimum solutions for a given \kDSA instance while only one of them is realizable, see \autoref{fig:exampleMultipleMinSolutionsNotAllRealizable} for an example.
Hence, instead of just computing one minimum-size solution, we iterate through these minimum-size solutions until one solution is realizable or \emph{all} solutions are tested.
Observe that there might be exponentially many minimum-size solutions:
In the \blockSeq $\B =\{ 0,3,1,3,1,\ldots,3,1,3\}$, for~$k=2$, each subsequence $3,1,3$ can be either changed to $2,2,3$ or to~$3,0,4$.
%
We use a data reduction rule (see end of this subsection) to reduce the amount of considered solutions~in~such~instances.

\smallskip\noindent{\bf \boldmath Criteria on the Realizability of \kDSA Solutions.}
A difficulty in the solutions provided by Phase~1, encountered in our preliminary experiments and as already observed by \citet{LSB12} on a real-world network, is the following:
If a solution increases the degree of one vertex~$v$ by some amount, say 100, and the overall number of vertices with increased degree is at most~$100$, then there are not enough neighbors for~$v$ to realize the solution.
We overcome this difficulty as follows: For a \kDSA-instance~$(\B,k)$ and a corresponding solution~$\B'$, let~$S$ be a \kInsertSet for~$G$ such that the block sequence of~$G+S$ is~$\B'$.
By definition, the block sequence of the graph induced by the edges~$S$ is~$\B'\circleddash\B$.
Hence, it is a necessary condition (for success in Phase~2) that $\B'\circleddash \B$ is a \emph{realizable} block sequence, that is, there is a graph with block sequence~$\B'\circleddash\B$.
\citet{TV03} have shown that it is enough to check to following \emph{Erd\H{o}s-Gallai characterization} of realizable degree sequence just once for each block.

\begin{lemma}[\cite{EG60}]\label{lem:erdos-gallai}
	Let $\D =\{ d_1, \ldots, d_n\}$ be a degree sequence sorted in descending order.
	Then~$\D$ is realizable if and only if $\sum_{i=1}^n d_i$ is even and for each $1 \le r \le n-1$ it holds that
	\begin{align}
		\sum_{i=1}^{r} d_i \le r (r-1) + \sum_{i = r+1}^{n} \min(r, d_i). \label{eq:erdos-gallai}
	\end{align}
\end{lemma}
 We call the characterization provided by \autoref{lem:erdos-gallai} the \emph{Erd\H{o}s-Gallai test}. Unfortunately, there are \kAnonymous sequences~$\D'$, passing the  Erd\H{o}s-Gallai test, but still or not realizable in the input graph~$G$ (see \autoref{fig:exampleMultipleMinSolutionsNotAllRealizable} for an example).

We thus designed an advanced version of the  Erd\H{o}s-Gallai test that takes the structure of the input graph into account.
To explain the basic idea behind, we first discuss how \cref{eq:erdos-gallai} in~\autoref{lem:erdos-gallai} can be interpreted:
Let $V^r$ be the set of vertices corresponding to the first $r$~degrees.
The left-hand side sums over the degrees of all vertices in~$V^r$.
This amount has to be at most as large as the number of edges (counting each twice) that can be ``obtained'' by making~$V^r$ a clique ($r(r-1)$) and the maximum number of edges to the vertices in~$V \setminus V^r$ (a degree\dash{}$d_i$ vertex has at most~$\min\{d_i,r\}$ neighbors in~$V^r$).
The reason why the Erd\H{o}s-Gallai test might not be sufficient to determine whether a sequence can be realized in~$G$ is that it ignores the fact that same vertices in~$V^r$ might be already adjacent in~$G$ and it also ignores the edges between vertices in~$V^r$ and $V \setminus {V^r}$.
Hence, the basic idea of our \emph{advanced Erd\H{o}s-Gallai test} is, whenever some of the vertices
corresponding to the degrees
can be uniquely determined, to subtract the corresponding number of edges as they cannot contribute to the right-hand side of \cref{eq:erdos-gallai}.

While the difference between using just the  Erd\H{o}s-Gallai test and the advanced Erd\H{o}s-Gallai test resulted in rather small differences for the lower bound (at most 10 edges), this small difference was important for some of our instances to succeed in Phase~2 and to optimally solve the instance.
We think that further improving the advanced Erd\H{o}s-Gallai test is the best way to improve the rate of success in Phase~2.

\smallskip\noindent{\bf Complete Strategy for Phase~1.}
With the above described restriction for realizable $k$\dash{}anonymous degree sequences, we finally arrive at the following problem for Phase~1, stated in the optimization form we solve:
\optprob{\textsc{\kDegSeqRealAnon} (\kDSRA)}
	{A degree sequence $\B$ and an integer~$k\in \N$.}
	{Compute all \kAnonymous degree sequences $\B'$ such that
$\B'\ogreaterthan\B$, $\|\B'\circleddash\B\|$ is minimum, and $\B'\circleddash\B$ is realizable.
}
Our strategy to solve \kDSRA is to iterate (for increasing solution size) through the solutions of \kDSA and run for each of them the advanced Erd\H{o}s-Gallai test.
Thus, we step-wisely increase the respective lower bound $\B'-\B$ until we arrive at some~$\B'$ passing the test.
Then, for each solution of this size we test in~Phase~2 whether it is realizable (if so, then we found an optimal solution).
If the realization in Phase~2 fails, then, for each such \blockSeq~$\B'$, we compute how many degrees have to be ``wasted'' in order to get a realizable sequence.
Wasting means to greedily increase some degrees in~$\B'$ (while preserving $k$\dash{}anonymity) until the resulting degree sequence is realizable in the input graph.
The cost~$\B'-\B$ plus the amount of degrees needed to waste in order to realize~$\B'$ is stored as an upper-bound.
A minimum upper-bound computed in this way is the result of~our~heuristic.

Due to the power law degree distribution in social networks, the degree of most of the vertices is close to the average degree, thus one typically finds in such instances two large blocks~$B_i$ and~$B_{i+1}$  containing many thousands of vertices.
Hence, ``wasting'' edges is easy to achieve by increasing degrees from~$B_i$ by one to~$B_{i+1}$ (this is optimal with respect to the Erd\H{o}s-Gallai characterization).
For the case that two such blocks cannot be found, as a fallback we also implemented a straightforward dynamic programming to find all possibilities to waste edges to obtain a realizable sequence.

\medskip
\emph{Remark.}
	We do not know whether the decision version of \kDSRA (find only one such solution~$\B'$) is polynomial-time solvable and resolving this question remains as challenge for future research.

\smallskip\noindent{\bf \boldmath Data Reduction Rule.}
In our preliminary experiments we observed that for some instances we could not finish Phase~1 even for~$k=2$ within a time limit of one hour.
Our investigations revealed that this is mainly due to the frequent occurrence of the following ``pattern'' within three consecutive blocks: The first block~$i$ and the third block~$i+2$ are each of size at least~$2k-1$ and the middle block consists of~$k-1$ degrees. For example, for $k=2$ consider the consecutive blocks~$4,1,4$.
The details of the dynamic program (see \autoref{ssect:step1}) show that in any solution for the entire block sequence block~$i$ and~$i+2$ stay full and either block~$i+1$ is filled by the degrees from block~$i$ or they are increased to block~$i+2$. In our example this means that the solution is either $4,0,4$ or $3,2,4$ (block~$i+2$ could contain less degrees but at least~$k$).
Then, if this pattern occurs multiple, say~$x$, times, then there~$2^x$ different solutions in Phase~1.
However, as can be observed in our example, with respect to the Erd\H{o}s-Gallai test whether the resulting sequence is realizable both solutions, $4,0,4$ and~$3,2,4$, are equivalent because they increase just one degree by one.
The general idea of our data reduction rule is to find these patterns where the first and last block are ``large'' enough to guarantee that the degrees of preceding blocks are not increased to the middle of the pattern and it is not necessary to increase something from the middle of the pattern to succeeding blocks.
Hence, the middle of the pattern can be solved ``independently'' from preceding and succeeding blocks and if there is a minimum solution which is ``Erd\H{o}s-Gallai-optimal'' (increasing degrees by at most one), then it is safe to take one of them for the middle of the pattern.
Formally, our data reduction rule, generalizing the above ideas, is as follows.

\begin{samepage}\begin{rrule}\label{rule:data-rule}
Let $(\B,k)$ be an instance of~\kDSRA.
 If there is a block~$B_i$ in~$\B$ with $b_i\ge k$, a sequence of blocks $B_j,B_j+1,\ldots,B_{j+t}$ such that $\sum_{l=j}^{j+t} b_l\ge (t+1)k+ k-1$ and $b_l\ge k$ for all $l\in \{j,\ldots,j+t\}$, and if there is minimum-size \kAnonymization~$\B'_{i,j}$ of the block sequence~$B_{i,j}=B_i,B_{i+1},\ldots,B_j$ such that i)~all blocks~$l$ for $l\ge 2$ in $\B'_{i,j}\circleddash \B_{i,j}$ are empty and~ii)~block zero in $\B'_{i,j}$ is of size at least~$k$, then substitute in~$\B$ the subsequence~$\B_{i,j}$ by~$\B'_{i,j}$.
\end{rrule}
\end{samepage}
In our implementation of \autoref{rule:data-rule} we use our dynamic programming algorithm (with disabled  Erd\H{o}s-Gallai test) to check whether there is a \kAnonymization~$\B'_{i,j}$ for~$\B_{i,j}$ fulfilling the required properties.

\subsection{Phase~2: Realizing a $k$-Anonymous Degree Sequence}\label{ssec:phase-2-des}

Let $(G,k)$ be an instance of \kDegAnon and let~$\B$ be the block sequence of~$G$.
In Phase~1 a \kAnonymization~$\B'$ of~$\B$~is computed such that $\B'\ogreaterthan \B$.
In Phase~2, given~$G$ and~$\B'$, the task is to decide whether there is a set~$S$ of edge insertions for~$G$ such that the block sequence of~$G+S$ is equal to~$\B'$. We call this the \textsc{\RealProb} problem and first prove that it is NP-hard.

\begin{theorem}\label{thm:realprob-np-hard}
\RealProb is NP-hard even on cubic planar graphs.
\end{theorem}
{\begin{proof}
	We prove the NP-hardness by a reduction from the \IndSet problem:
	Given a graph~$G=(V,E)$ and an integer~$\size$, decide whether there is a set of at least~$\size$ pairwise non-adjacent vertices in~$G$.
	\IndSet remains NP-hard in cubic planar graphs~\cite[GT20]{GJ79}.

	The reduction, which is similar to those proving that \kDegAnon remains NP-hard on 3-colorable graphs~\cite[Theorem~1]{HNNS13}, is as follows:
	Let~$G$ be a cubic planar and~$\size$ an integer that together form an instance of \IndSet.
	The block sequence of the $n$-vertex graph~$G$ is $\B=\{0,0,0,n\}$.
	We~set \[\B'=\{0,0,0,n-\size,\underbrace{0,\ldots,0}_{\size-2},\size\}\]
	and next prove that the \RealProb-instance $(G,\B')$ is a yes-instance if and only if $(G,\size)$ is a yes-instance for \IndSet.

	If there is an independent set~$S$ (pairwisely non-adjacent vertices) of size~$\size$ in~$G$, then adding all edges between the vertices in~$S$ (making them a clique), results in a graph whose block sequence in~$\B'$.
	Reversely, in a realization of~$\B'$ in~$G$, there are exactly~$\size$ vertices whose degree has been increased by~$\size-1$. Hence, these vertices form a clique, implying that they are independent in~$G$.

	We remark that from the details of the proof of Theorem~1~\cite{HNNS13}, it follows that \RealProb is NP-hard even in case that~$\B'$ is a \kAnonymized sequence such that $\|\B'-\B\|$ is minimum.\qed
\end{proof}}
%
We next present our heuristics for solving \RealProb.
First, we find a degree-vertex mapping, that is, for~$\D'=d'_1,\ldots,d'_n$ being the degree sequence corresponding to~$\B'$, we assign each value~$d'_i$ to a vertex~$v$ in~$G$ such that~$d'_i\ge \deg_G(v)$ and set~$\dem(v)$, the demand of~$v$, to~$d'_i-\deg_G(v)$.
Second, we try to find, mainly by the local exchange heuristic, an edge insertion set~$S$ such that in~$G+S$ the amount of incident new edges for each vertex~$v$ is equal to its demand~$\dem(v)$.
The details in the proof of \autoref{thm:realprob-np-hard} indeed show that already finding a realizable degree-vertex mapping is NP-hard.
This coincides with our experiments, as there the ``hard part'' is to find a good degree-vertex mapping and the local exchange heuristic is quite successful in realizing it (if possible). Indeed, we prove that ``large'' solutions can be always realized~by~it:

\begin{theorem}
A demand function~$\dem$ is always realizable by the local exchange heuristic in a maximum degree-$\Delta$ graph~$G=(V,E)$ if $\sum_{v\in V}\dem(v)\ge 20\Delta^4+4\Delta^2$.
\end{theorem}
In \autoref{ssec:detailed-phase-2-des} (appendix) we give a detailed description of our algorithms for~Phase~2 and formally prove the above theorem.


\section{Experimental Results}\label{sec:experiments}

\smallskip\noindent{\bf Implementation Setup.}
All our experiments are performed on an Intel Xeon E5-1620 3.6GHz machine with 64GB memory under the Debian GNU/Linux 6.0 operating system.
The program is implemented in Java and runs under the OpenJDK runtime environment in version~1.7.0\_25.
The time limit for one instance is set to one hour per~$k$-value and we tested for~$k ={}$2, 3, 4, 5, 7, 10, 15, 20, 30, 50, 100, 150, 200.
After reaching the time limit, the program is aborted and the upper and lower bounds computed so far by the dynamic program for Phase~1 are returned.
The source code is freely available.\footnote{\url{http://fpt.akt.tu-berlin.de/kAnon/}}

\smallskip\noindent{\bf Real-World Instances.}
We considered the five social networks from the co-author citation category in the 10$^\text{th}$ DIMACS challenge~\cite{Dim12}.
We compared the results of our upper bounds against an implementation of the clustering-heuristic provided by \citet{LSB12} and against the lower bounds given by the dynamic program.
Our algorithm could solve 26\% of the instances to optimality within one hour.
Interestingly, our exact approach worked best with the coPapersCiteseer graph from the 10$^\text{th}$ DIMACS challenge although this graph was the largest one considered (in terms of~$n+m$).
For all tested values of~$k$ except~$k=2$, we could optimally $k$\dash{}anonymize this graph and for~$k=2$ our upper bound heuristic is just two edges away from our lower bound.
The coAuthorsDBLP graph is a good representative for the results on the DIMACS-graphs, see \Cref{tbl:dimacSmallPart}:
\begin{table}[t]
	\renewcommand{\tabcolsep}{2pt}
	\caption{Experimental results on real-world instances. 
	We use the following abbreviations: CH for clustering-heuristic of \citet{LSB12}, OH for our upper bound heuristic, OPT for optimal value for the \kDegAnon problem, and DP for dynamic program for the \kDSRA problem.
	If the time entry for DP is empty, then we could not solve the \kDSRA instance within one-hour and the DP bounds display the lower and upper bounds computed so far.
	If OPT is empty, then either the \kDSRA solutions could not be realized or the \kDSRA instance could not be solved within one hour.
	} \label{tbl:dimacSmallPart}
	\begin{tabular}{cr|rrr|rr|rrr}
&  & \multicolumn{3}{c|}{solution size} & \multicolumn{2}{c|}{DP bounds} & \multicolumn{3}{c}{time (in seconds)} \\
graph & k & CH & OH & OPT & lower & upper & CH & OH & DP \\ \thickhline
\multirow{4}{2.5cm}{coAuthorsDBLP\\ $(n\approx2.9\cdot 10^5$,\\ $m\approx9.7\cdot 10^5$,\\$\Delta=336)$} & 2 & 97 & 62 & \emptyBox & 61 & 61 & 1.47 & 0.08 & 0.043 \\
  & 5 & 531 & 321 & 317 & 317 & 317 & 1.41 & 0.29 & 26.774 \\
  & 10 & 1,372 & 893 & \emptyBox & 869 & 869 & 1.03 & 0.48 & 1.58 \\
  & 100 & 21,267 & 15,050 & \emptyBox & 10,577 & 11,981 & 1.13 & 885.79 & \emptyBox \\ \hline
\multirow{4}{2.5cm}{coPapersCiteseer\\ $(n\approx4.3\cdot 10^5$,\\ $m\approx1.6\cdot 10^7$,\\$\Delta=1188)$} & 2 & 203 & 80 & \emptyBox & 78 & 78 & 9.9 & 0.1 & 0.394 \\
  & 5 & 998 & 327 & 327 & 327 & 327 & 10.32 & 0.19 & 0.166 \\
  & 10 & 2,533 & 960 & 960 & 960 & 960 & 8.83 & 0.74 & 0.718 \\
  & 100 & 51,456 & 22,030 & 22,007 & 22,007 & 22,007 & 5.97 & 263.95 & 264.553 \\ \hline
\multirow{4}{2.5cm}{coPapersDBLP\\ $(n\approx5.4\cdot 10^5$,\\ $m\approx1.5\cdot 10^7$,\\$\Delta=3299)$} & 2 & 1,890 & 1,747 & \emptyBox & 950 & 1,733 & 11.28 & 2.13 & \emptyBox \\
  & 5 & 9,085 & 8,219 & \emptyBox & 4,414 & 8,121 & 10.66 & 28.83 & \emptyBox \\
  & 10 & 19,631 & 17,571 & \emptyBox & 9,557 & 17,328 & 9.95 & 149.56 & \emptyBox \\
  & 100 & 258,230 & \emptyBox & \emptyBox & 128,143 & 233,508 & 22.16 & \emptyBox & \emptyBox \\\thickhline
	\end{tabular}
\end{table}
A few instances could be solved optimally and for the remaining ones our heuristic provides a fairly good upper bound.
One can also see that the running times of our algorithms increase (in general) exponentially in~$k$.
This behavior captures the fact that our dynamic program for Phase~1 iterates over all minimal solutions and for increasing~$k$ the number of these solutions increases dramatically.
Our heuristic also suffers from the following effect:
Whereas the maximum running time of the clustering-heuristic heuristic was one minute, our heuristic could solve 74\% of the instances within one minute and did not finish within the one-hour time limit for 12\% of the tested instances.
However, the solutions produced by our upper bound heuristic are always smaller than the solutions provided by the clustering-heuristic, on average the clustering-heuristic results are 72\% larger than the results of our heuristic.

\smallskip\noindent{\bf Random Instances.}
We generated random graphs according to the model by Barabási–Albert~\cite{BA99} using the implementation provided by the AGAPE project~\cite{BLL12} with the JUNG library\footnote{\url{http://jung.sourceforge.net/}}.
Starting with~$m_0=3$ and~$m_0 = 5$ vertices these networks evolve in~$t \in \{400, 800, 1200, \ldots, 34000\}$ steps.
In each step a vertex is added and made adjacent to~$m_0$ existing vertices where vertices with higher degree have a higher probability of being selected as neighbor of the new vertex.
In total, we created 170 random instances.

Our experiments reveal that the synthetic instances are particular hard.
For example, even for~$k=2$ and~$k=3$ we could only solve 14\% of the instances optimal although our dynamic program produces solutions for Phase~1 in 96\% of the instances.
For higher values of~$k$ the results are even worse (for example zero exactly solved instances for~$k=10$).
This indicates that the current lower bound provided by Phase~1 needs further improvements.
However, the upper bound provided by our heuristic are not far away: On average the upper bound is 3.6\% larger than the lower bound and the maximum is 15\%.
Further enhancing the advanced Erd\H{o}s-Gallai test seem to be the most promising step towards closing this gap between lower and upper bound.
Comparing our heuristic with the clustering-heuristic reveal similar results as for real-world instances.
Our heuristic always beats the clustering-heuristic in terms of solution size, see \Cref{fig:comparisonHeuristicsRandomGraphs} for~$k=2$ and~$k=3$.
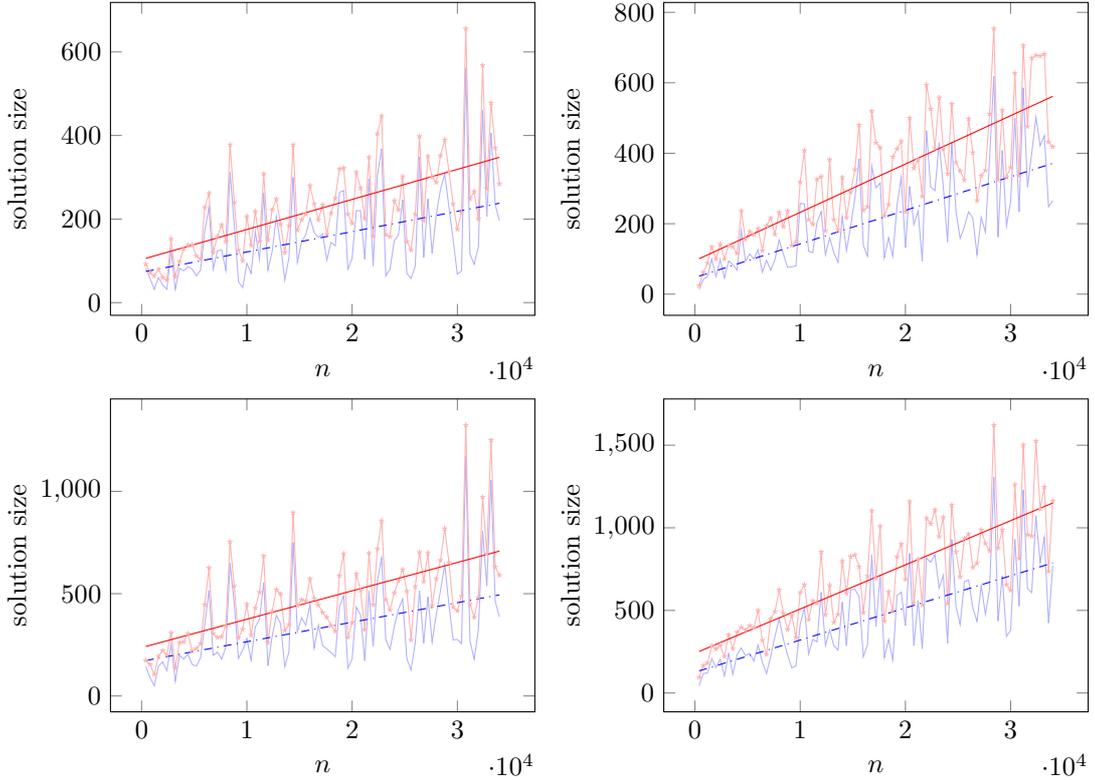
\begin{figure}[t]
	\centering
	\tikzset{every mark/.append style={scale=0.5}}
	\begin{subfigure}[t]{0.48\textwidth}
		\begin{tikzpicture}[scale=1]
			\begin{axis}[
						width=\textwidth,
						height=0.8\textwidth,
						xlabel={$n$},
						ylabel={solution size},
						legend cell align=left,
						legend pos=north west]

				\addplot[color=blue!30] table [x={t}, y={our_upper}] {synthetic_k_2_n_3.dat};
				\addplot[color=red!30,mark=star] table [x={t}, y={Clustering_Heuristic_size}] {synthetic_k_2_n_3.dat};
				\addplot[color=blue,dashdotted] table [x={t}, y={create col/linear regression={y={our_upper}}}] {synthetic_k_2_n_3.dat};
				\addplot[color=red] table [x={t}, y={create col/linear regression={y={Clustering_Heuristic_size}}}] {synthetic_k_2_n_3.dat};
			\end{axis}
		\end{tikzpicture}
	\end{subfigure}
	\begin{subfigure}[t]{0.48\textwidth}
		\begin{tikzpicture}[scale=1]
			\begin{axis}[
						width=\textwidth,
						height=0.8\textwidth,
						xlabel={$n$},
						ylabel={solution size},
						legend cell align=left,
						legend pos=north west]

				\addplot[color=blue!30] table [x={t}, y={our_upper}] {synthetic_k_2_n_5.dat};
				\addplot[color=red!30,mark=star] table [x={t}, y={Clustering_Heuristic_size}] {synthetic_k_2_n_5.dat};
				\addplot[color=blue,dashdotted] table [x={t}, y={create col/linear regression={y={our_upper}}}] {synthetic_k_2_n_5.dat};
				\addplot[color=red] table [x={t}, y={create col/linear regression={y={Clustering_Heuristic_size}}}] {synthetic_k_2_n_5.dat};
			\end{axis}
		\end{tikzpicture}
	\end{subfigure}

	\vspace{0.15cm}

	\begin{subfigure}[b]{0.48\textwidth}
		\begin{tikzpicture}[scale=1]
			\begin{axis}[
						width=\textwidth,
						height=0.8\textwidth,
						xlabel={$n$},
						ylabel={solution size},
						legend cell align=left,
						legend pos=north west]
				\addplot[color=blue!30] table [x={t}, y={our_upper}] {synthetic_k_3_n_3.dat};
				\addplot[color=red!30,mark=star] table [x={t}, y={Clustering_Heuristic_size}] {synthetic_k_3_n_3.dat};
				\addplot[color=blue,dashdotted] table [x={t}, y={create col/linear regression={y={our_upper}}}] {synthetic_k_3_n_3.dat};
				\addplot[color=red] table [x={t}, y={create col/linear regression={y={Clustering_Heuristic_size}}}] {synthetic_k_3_n_3.dat};
			\end{axis}
		\end{tikzpicture}
	\end{subfigure}
	\begin{subfigure}[b]{0.48\textwidth}
		\begin{tikzpicture}[scale=1]
			\begin{axis}[
						width=\textwidth,
						height=0.8\textwidth,
						xlabel={$n$},
						ylabel={solution size},
						legend cell align=left,
						legend pos=north west]
				\addplot[color=blue!30] table [x={t}, y={our_upper}] {synthetic_k_3_n_5.dat};
				\addplot[color=red!30,mark=star] table [x={t}, y={Clustering_Heuristic_size}] {synthetic_k_3_n_5.dat};
				\addplot[color=blue,dashdotted] table [x={t}, y={create col/linear regression={y={our_upper}}}] {synthetic_k_3_n_5.dat};
				\addplot[color=red] table [x={t}, y={create col/linear regression={y={Clustering_Heuristic_size}}}] {synthetic_k_3_n_5.dat};
			\end{axis}
		\end{tikzpicture}
	\end{subfigure}
	\caption{Comparison of our heuristic (always the light blue line without marks) with the clustering-heuristic (always the light red line with little star as marks) on random data with different parameters: Top row is for~$k=2$, bottom row for~$k=3$; the left column is for~$m_0 = 3$, and the right column for~$m_0=5$. The linear, solid dark red line and dash-dotted blue line are linear regressions of the corresponding data plot. One can see that our heuristic produces always smaller~solutions.} \label{fig:comparisonHeuristicsRandomGraphs}
\end{figure}
We remark that for larger values of~$k$ the running time of the heuristic increases dramatically:
For~$k=30$ our algorithm provides upper bounds for 96\% of the instances, whereas for~$k=150$ this value drops to 18\%.

\section{Conclusion}
We have demonstrated that our algorithm framework is suitable to solve \kDegAnon on real-world social networks.
The key ingredients for this is an improved dynamic programming for the task to $k$\dash{}anonymize degree sequences together with certain lower bound techniques, namely the advanced Erd\H{o}s-Gallai test.
We have also demonstrated that the local exchange heuristic due to \citet{LT08} is a powerful algorithm for realizing \kAnonymous sequences and provided some theoretical justification for this effect.

The most promising approach to speedup our algorithm and to overcome its limitations on the considered random data, is to improve the lower bounds provided by the advanced Erd\H{o}s-Gallai test.
Towards this, and also to improve the respective running times, one should try to answer the question whether one can find in polynomial-time a minimum \kAnonymization~$\D'$ of a given degree sequence~$\D$ such that $\D'-\D$ is realizable.

{ 
\bibliographystyle{abbrvnat}
\bibliography{bibliography}
}
\appendix
\newpage

\section{Detailed Description of the Algorithms for~Phase~2: Realizing the $k$-anonymous degree sequence}\label{ssec:detailed-phase-2-des}

\smallskip\noindent{\bf Phase~2.1: Finding a Degree-Vertex Mapping.}
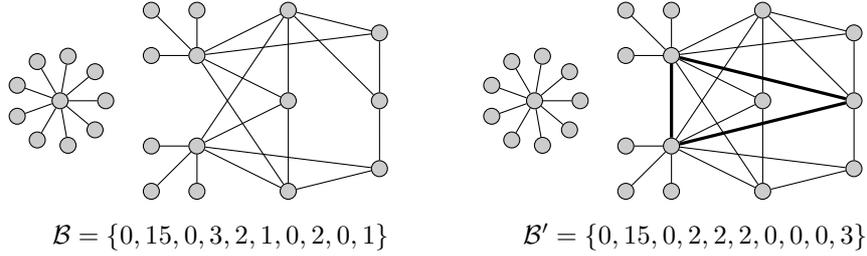
\begin{figure}[t]
	\centering
		\tikzstyle{vertex}=[circle,draw=black,minimum size=6pt,inner sep=0.8pt,fill=black!20]
		\def\n{9}
		\begin{tikzpicture}[scale=1.2]
			\foreach \i in {1,...,\n} {
				\node[vertex] (v-\i) at ( { 1 + 0.5 * cos(360 * \i / \n ) } , {0.5 * sin(360 * \i / \n )} ) {};
			}
			\foreach[count=\i] \x/\y in {2.5/-1,2/-1,2/-0.5, 2.5/1,2/1,2/0.5}
				\node[vertex](v2-\i) at (\x,\y) {};
			\foreach[count=\i] \x/\y in {1/0,2.5/-0.5,2.5/0.5,3.5/-1,3.5/0,3.5/1,4.5/-0.75,4.5/0,4.5/0.75}
				\node[vertex](v\i) at (\x,\y) {};
			\node at (2.75,-1.5) {$\B =\{ 0, 15, 0, 3, 2, 1, 0, 2, 0, 1\}$};
			\foreach \from in {1,...,\n}
				\draw (v-\from) -- (v1);
			\foreach \from in {1,...,3}
				\draw (v2-\from) -- (v2);
			\foreach \from in {4,...,6}
				\draw (v2-\from) -- (v3);
			\draw (v2) -- (v7);
			\draw (v3) -- (v9);

			\foreach \from/\to in {4/5,4/7,5/6,6/8,6/9,2/4,2/5,2/6,3/4,3/5,3/6,7/8,8/9}
				\draw (v\from) -- (v\to);

			\begin{scope}[xshift=5.2cm]
				\foreach \i in {1,...,\n} {
					\node[vertex] (v-\i) at ( { 1 + 0.5 * cos(360 * \i / \n ) } , {0.5 * sin(360 * \i / \n )} ) {};
				}
				\foreach[count=\i] \x/\y in {2.5/-1,2/-1,2/-0.5, 2.5/1,2/1,2/0.5}
					\node[vertex](v2-\i) at (\x,\y) {};
				\foreach[count=\i] \x/\y in {1/0,2.5/-0.5,2.5/0.5,3.5/-1,3.5/0,3.5/1,4.5/-0.75,4.5/0,4.5/0.75 }
					\node[vertex](v\i) at (\x,\y) {};
				\node at (2.75,-1.5) {$\B' =\{ 0, 15, 0, 2, 2, 2, 0, 0, 0, 3\}$};
				\foreach \from in {1,...,\n}
					\draw (v-\from) -- (v1);
				\foreach \from in {1,...,3}
					\draw (v2-\from) -- (v2);
				\foreach \from in {4,...,6}
					\draw (v2-\from) -- (v3);
				\draw (v2) -- (v7);
				\draw (v3) -- (v9);

				\foreach \from/\to in {4/5,4/7,5/6,6/8,6/9,2/4,2/5,2/6,3/4,3/5,3/6,7/8,8/9}
					\draw (v\from) -- (v\to);
				\draw[very thick] (v2) -- (v8);
				\draw[very thick] (v3) -- (v8);
				\draw[very thick] (v2) -- (v3);
			\end{scope}
		\end{tikzpicture}
	\caption{Our smallest example where ``jumps'' are necessary to obtain a minimum-size solution. The only minimum solution to 2-anonymize the left graph is to add the bold edges in the right graph. Observe that, although there are two degree-four vertices in the graph, the solution lifts a degree-three vertex to degree five, that is, there is a ``jump'' of two.
	}
	\label{fig:jumps-are-necessary}
\end{figure}
Given a graph~$G$ with its block sequence~$\B$ and a \kAnonymous block sequence~$\B'\ogreaterthan\B$, the two main difficulties that arise when trying to find a best possible (realizable) degree-vertex mapping for~$\B'$ are as follows:
To illustrate to first one, consider to  2\dash{}anonymize a graph consisting of two connected components $\{a,b,c\}$ and $\{d,e\}$ where each component is just a path.
Hence, the block sequence is $\B=\{0,4,1\}$ and a minimum solution would be to insert an edge between two degree-one vertices, resulting in the block sequence $\B'=\{0,2,3\}$.
Given~$\B'$, a degree-vertex mapping has to choose two degree-one vertices where all but the choice~$\{d,e\}$ leads to a realization.
Hence, the basic problem is that a degree-vertex mapping has to choose $x$~many vertices from~block~$i$ which is of size more than~$x$ and thus the assignment is non-unique.
In our experiments we observed that this difficulty can be
solved satisfactorily by randomly selecting the vertices from the blocks.
However, the second difficulty is a more severe problem, also on practical instances.

Assume that $\B=\{3,2,1\}$ is the block sequence of our input graph (three degree-zero vertices and a path of length~two) and the result of Phase~1 is the \kAnonymized[2] block sequence~$\B'=\{2,2,2\}$.
Now, the difficulty arises that there are actually two ``interpretations'' of~$\B'$: The first (natural) one would be
to increase a degree zero up to~one and a degree two up to~three.
However, the second would be that one degree zero is increased by two up to three.
We call this a \emph{jump} since a degree is increased ``over'' a non-empty block, while the natural interpretation (making most sense in the majority of the cases) is that a degree is increased to the next non-empty block~$B$ and from there, first the vertices originally in~$B$ are increased further.
While in the example above the second ``jump''-interpretation cannot be realized (only one vertex has non-zero demand),
\autoref{fig:jumps-are-necessary} illustrates an example where the only realizable degree-vertex mapping has such a jump.

In our experiments, against our a-priori intuition, we observed that the \kAnonymized sequences~$\B'$ (computed in Phase~1) have typically less than ten ``jump blocks'' (a jump over these blocks is possible) and for each of these blocks up to five degrees can jump ``over'' it.
Since the number of jump blocks is reasonably small and as we try to realize many degree-vertex mappings for each~$\B'$ (100~in the results presented in \autoref{sec:experiments}), it turned out that, for increasing~$\alpha$ from~zero up to the number of jump blocks, iterating through all possibilities to choose~$\alpha$ jump blocks is a good choice.
Having fixed the jumps, it follows how many degrees from~$i$ are increased to~$j$ and we randomly (25~trials for each jump configuration) select the appropriate number of vertices from block~$B_i$ in~$G$.
In total, for each given~$\B'$ from Phase~1 we try to realize $25\cdot 100=2500$ degree-vertex mappings.
These parameters ($25$ and~$100$) have been chosen according to the results in preliminary experiments and seem to be a good compromise between expected success rate and the needed time.

\smallskip\noindent{\bf Phase~2.2: Realizing a Degree-Vertex Mapping.}
In the last part of finding a realization of a \kAnonymized sequence~$\B'$ in a graph~$G=(V,E)$, one is given a degree-vertex mapping which provides a non-negative integer demand for each vertex and the task is to decide whether it is realizable, that is, is there an edge insertion set~$S$ such that in $G+S$ the amount of incident new edges for each vertex  is equal to its demand.
Formally, let $\dem\colon V\rightarrow\N$ be the function providing the demand of each vertex.
Whether~\dem is realizable can be decided in polynomial-time by solving an $f$\dash{}factor instance and it has been shown that, for the maximum degree~$\Delta$ of~$G$, $\dem$~is always realizable if $\sum_{v\in V} \dem(v)\ge (\Delta^2+4\Delta+3)^2$~\cite[Lemma~4]{HNNS13}.
We have implemented a so-called local exchange heuristic by~\citet{LT08} which turned out to perform surprisingly well.
Indeed, we here present also some theoretical justification for this, formally proving that basically the same lower bound (as for $f$\dash{}factor) on $\sum_{v\in V}\dem(v)$ is enough to guarantee that the local exchange heuristic always realizes~$\dem$.

In principle, the local exchange heuristic adds edges between vertices as long as possible to satisfy their demand and if it gets stuck at some point, then it tries to continue by exchanging an already inserted edge.
Formally, it works as follows:
Let~$S$ be the set of new edges which is initialized by~$\emptyset$.
As long as there are two vertices~$u$ and~$v$ with non-zero demand, check whether the edge $\{u,v\}$ is insertable, meaning that neither $\{u,v\}\in E$ nor $\{u,v\}\in S$. If it is insertable, then add $\{u,v\}$ to~$S$ and decrease the demand of~$u$ and~$v$ by one.
If this procedure ends with all vertices having demand zero, then $S$~is an insertion set realizing~$\dem$.
Otherwise, we are left with a set~$V^\dem$ of vertices with non-zero demand.
If there are two vertices~$v_1,v_2\in V^\dem$, then for each edge $\{u,w\}\in S$ check whether the two edges $\{v_1,u\}$ and~$\{v_2,w\}$ or $\{v_1,w\}$ and~$\{v_2,v\}$ are insertable. If so, then delete~$\{u,w\}$ from~$S$, insert the two edges that are insertable, and decrease the demand of $v_1$ and~$v_2$ by one.
In the special case of $V^\dem$ containing only one vertex~$v$, then it holds that the remaining demand of~$v$ is at least two, because $\sum_{v\in V}\dem(v)$ can be assumed to be even (otherwise it is not realizable).
In this case perform the following for each edge $\{u,w\}\in S$: Check whether $\{v,u\}$ and $\{v,w\}$ are insertable and if so, then insert them to~$S$, delete~$\{u,w\}$ from~$S$, and decrease the demand of~$v$ by two.

We have implemented the local exchange heuristic so that it first randomly tries to add edges and then, if stuck at some point, performs the above described exchange operations (if possible).
We conclude with proofing a certain lower bound on $\sum_{v\in V}\dem(v)$ which guarantees the success of the local exchange heuristic.
As a first step for this, we prove that any demand function can be assumed to require to increase the vertex degrees at most up to~$2\Delta^2$.

\begin{lemma}\label{lem:degreeBound}
 Any minimum-size \kInsertSet for an instance of \kDegAnon yields a graph with maximum degree at most~$\degreeBound$.
\end{lemma}
{
\begin{proof}\label{proof:degreebound}
	Before proving \autoref{lem:degreeBound}, we introduce the terms ``co-matching'' and ``co-cycle cover'' and prove an observation concerning their existence.
	A graph~$G = (V,E)$ contains a \emph{co-matching} of size~$\ell$ if~$\overline{G}$ contains a matching of size~$\ell$, that is, a subset of~$\ell$~non-overlapping edges of~$\overline{G}$.
	A \emph{perfect} co-matching of~$G$ is a co-matching of size~$\nicefrac{|V|}{2}$.
	Analogously, $G$ contains a \emph{co-cycle cover} if~$\overline{G}$ contains a cycle cover, that is, a subgraph of~$\overline{G}$ with~$|V|$ vertices such that each vertex has degree two.
	We prove the following observation that shows sufficient conditions for the existence of co-matchings and co-cycle covers.

	\begin{obs}\label{obs:perf-co-matching-cycle-cover}
		Let~$G = (V, E)$ be a graph and let $V'\subseteq V$ be a vertex subset with~$|V'| \ge 2\Delta + 1$. Then, $G[V']$ contains a co-cycle cover and if~$|V'|$ is even, then $G[V']$ contains a perfect co-matching.
	\end{obs}
	\begin{proof}
		Since~$|V'| \ge 2\Delta + 1$, it follows that in~$\overline{G[V']}$ every vertex has degree at least~$|V'|-\Delta \ge \nicefrac{|V'|}{2}$.
		Hence, using Diracs Theorem~\cite{Die10}, it follows that~$\overline{G[V']}$ contains a Hamiltonian cycle~$C$. Thus it contains a co-cycle cover.
		Additionally, if~$|V'|$ is even, then it follows that the number of vertices in~$C$ is even and, hence, taking every second edge of~$C$ results in a perfect matching.
	\end{proof}

	We now prove \autoref{lem:degreeBound}.
   We consider an instance $(G=(V,E),k)$ of \kDegAnon. Let~$S$~be a minimum-size \kInsertSet for~$G$.
	Suppose the maximum degree in~$G+S$ is greater than~$\degreeBound$.
	Let~$\TG$ be the graph that is obtained from~$G$ by iterating over all edges in~$S$ (in an arbitrary order) and adding an edge if it does not increase the maximum degree~$\Delta$. 
	Let~$\TS\subseteq S$ be the edges not contained in~$\TG$. By definition~$\TG+\TS$ is a $k$-anonymous graph and adding any edge from~$\TS$ to~$\TG$ causes a maximum degree of~$\Delta+1$.
	Hence, denoting by~$X\subseteq V(\TS)$ the vertices of degree~$\Delta$, each edge in~$\TS$ has at least one endpoint in~$X$.  Let~$Z=V(\TS)\setminus X$.
	Clearly, the vertices in~$\TG+\TS$ that have degree greater than~$\Delta$ are a subset of $V(\TS)=X\cup Z$. Next, we show how to construct an edge set of size less than~$|\TS|$ whose addition to~$\TG$ results in a $k$-anonymous graph.

     \newcommand{\xSplit}{(\Delta+3)\Delta}
	\textbf{Case~1:} $|X|\le \xSplit$. \\
	Since, each edge in~$\TS$ contains at least one endpoint from~$X$, every vertex in~$Z$ is incident to at most~$\xSplit$ edges in~$\TS$.
	Hence, the degree of each vertex from~$Z$ in~$\TG+\TS$ is less than $\xSplit+\Delta \le \degreeBound$.
	Thus, only the vertices from~$X$ can have degree more than $\degreeBound$ in~$\TG+\TS$ and each of them is adjacent to at least~$\degreeBound-\xSplit > \Delta^2/2 + 1$ vertices from~$Z$.
	Now for each vertex~$u$ of maximum degree in~$\TG+\TS$ find two non-adjacent neighbors of it in~$Z$, delete the edges between~$u$ and the neighbors and insert an edge between the neighbors.
	(Observe that the neighbors always exist, since $\TG$ has maximum degree at most~$\Delta-1$ and each such vertex~$u$ has at least $\Delta^2/2+1 > \Delta-1$ neighbors from~$Z$.)
	Hence, we get a smaller set of edge additions that also transforms~$\TG$ into a $k$-anonymous graph, implying a contradiction.

  \textbf{Case 2: $|X|>\xSplit$}\\
    We give an algorithm that transforms~$\TG$ by inserting at most~$|\TS|$ edges into a $k$-anonymous graph.

    \begin{enumerate}
     \item Initialize~$S'$ by a copy of~$\TS$ and~$G'$ by a copy of~$\TG$.
     \item \label{op:clique}While there are any two vertices $v,u\in V(S')$ that are non-adjacent and both have degree less than~$\Delta$ in~$G'$, add the edge~$\{u,v\}$ to~$G'$, delete one edge in~$S'$ that is incident to~$v$ and one that is incident to~$u$.
     \item \label{op:comatching} Let~$Z'=\{u\in Z\cap V(S')\mid \deg_{G'}(u)<\Delta\}$. Find a subset of vertices~$Y\subseteq X$ such that there is co-matching in~$G'$ from~$Z'$ to~$Y$ such that each vertex in~$Y$ is contained in exactly one co-matching edge and each vertex $z\in Z'$ is contained in exactly $\min\{\Delta,\deg_{\TG+\TS}(z)\}-\deg_{G'}(z)$ co-matching edges. Add the edges of this co-matching.
    \end{enumerate}
    To complete the algorithm we distinguish between several cases.
	Before that, observe that after Step~\ref{op:clique} the maximum degree of~$G'$ is still~$\Delta$.
	Additionally, since each vertex in~$Y\subseteq X$ is adjacent to only one co-matching edge and each vertex~$z \in Z'$ gets only $\min\{\Delta,\deg_{\TG+\TS}(z)\}-\deg_{G'}(z)$ additional edges, after Step~\ref{op:comatching} the maximum degree of~$G'$ is~$\Delta+1$.
	As Step~\ref{op:clique} is exhaustively applied, $Z'$ induces a clique and thus $|Z'|\le \Delta$. This implies that $|Y|<\Delta^2$.
	The existence of the set~$Y$ follows from the fact that after Step~\ref{op:clique} in~$G'$ each vertex in~$Z'$ is adjacent to at most~$\Delta-1$ vertices in~$X$ and $|X|>\xSplit$.
	Thus for each vertex~$z$ one can pick any $\min\{\Delta,\deg_{\TG+\TS}(z)\}-\deg_{G'}(z)$ non-adjacent vertices for~$X$ that are disjoint from those chosen for the others.

	Additionally, since all vertices from~$X$ have degree larger than~$\Delta$ in~$\TG+\TS$, the degree of each vertex in~$G'$ is at most as its degree in~$\TG+\TS$.
	Denote by~$X'$ all vertices in~$\TG+\TS$ that have degree greater than~$\Delta$ but are not contained in~$Y$.
	All vertices in~$X'$ have degree~$\Delta$ in~$G'$ (see Step~\ref{op:comatching}).
	Clearly, $X'\cup Y$ are exactly the vertices in~$\TG+\TS$ that have degree greater than~$\Delta$ and thus $|X'\cup Y|\ge k$.
	Furthermore, we do not change the degree of any vertex~$V \setminus (X' \cup Y)$.
	Thus, in the following cases it remains to argue that none of the vertices in~$X' \cup Y$ damage the \kAnonymous constraint and that the new solution adds less edges than the old one.

    \textbf{Case~2.1: $|X'|$ is even}\\
    As $(X\setminus Y)\subseteq X'$ and $|Y|\le \Delta^2$ it follows that $|X'|>3\Delta$. Hence, by \autoref{obs:perf-co-matching-cycle-cover} there is a perfect co-matching on the vertices in~$X'$ and inserting the corresponding edges results in a $k$-anonymous graph. The number of corresponding edge additions is less than those in~$\TS$, because $\TS$ increases the degree of all vertices in~$X'\cup Y$ also at least to~$\Delta + 1$ and for some of them even above~$\degreeBound$.

    \textbf{Case~2.2: $|X'|$ is odd}\\
    \textbf{Case~2.2.1 $k<|X'\cup Y|/2$}\\
    We increase via a co-matching  on $2\lceil \frac{k-|Y|}{2}\rceil$ vertices from~$X'$ the degree of enough vertices to~$\Delta+1$ such that together with~$Y$ there are at least~$k$ degree-$(\Delta+1)$ vertices. Observe that, $|X'|-2\lceil \frac{k-|Y|}{2}\rceil>2k-|Y|-(k-|Y|)=k$ and, thus, there are at least~$k$ vertices left with degree~$\Delta$.

    \textbf{Case~2.2.2 $k\ge |X'\cup Y|/2$} \\
    \textbf{Case~2.2.2.1 $|Y|$ is even}
    In this case we add the edges of a co-matching on~$Y$ and the edges of a co-cycle in~$X'$. This results in a $k$-anonymous graph with maximum degree~$\Delta+2$. Note that the number of edge additions in this solution to get degree $\Delta +2$ for the vertices in~$X'\cup Y$ is less than $(\Delta+2)\cdot |X'\cup Y|$ whereas for $\TG+\TS$ the set $\TS$ contains at least
    \begin{align*}
\frac{k \cdot (\degreeBound-\Delta)}{2}\ge \frac{|X'\cup Y|\cdot (\degreeBound-\Delta)}{4}
    \end{align*}
    edges to increase for at least~$k$ vertices the degree from at most $\Delta$ to at least~$\degreeBound$.

    By \autoref{obs:perf-co-matching-cycle-cover},  there exists a co-cycle cover on~$X'$ since~$|X'|> 3\Delta$.
	It remains to argue that there is co-matching on~$Y$.
	This follows from \autoref{obs:perf-co-matching-cycle-cover} in case of~$|Y|>2\Delta$.
	However, in case of $|Y|\le 2\Delta$, because of $|X|>(\Delta+3)\Delta$ the choice of~$Y$ in Step~\ref{op:comatching} can be easily adjusted to guarantee the existence of such a matching.

    \textbf{Case~2.2.2.2 $|Y|$ is odd}\\
    Observe that from $k\ge |X'\cup Y|/2$ it follows that $\TG+\TS$ can contain at most two vertex-degrees that are larger than~$\Delta$. We first argue that there cannot be just one vertex-degree greater than~$\Delta$. Recall that $X'\cup Y$ are exactly the vertices in~$\TG+\TS$ the have degree larger than~$\Delta$, say they have degree~$a$. We show that the difference in the sum of the vertex degrees from~$\TG$ compared to~$\TG+\TS$ is an odd number, a contradiction.
    First, note that the number of edges that are deleted from $\TS$ in Step~\ref{op:clique} is an even number, we ignore them in the following.
    Second, $\TS$ also contains edges that increase (as in Step~\ref{op:comatching}) the degree of each vertex $z\in Z'$ from $\deg_{G'}(z)$ to $\min\{\Delta,\deg_{\TG+\TS}(z)\}$ and they contribute $2|Y|$ to the sum of the degrees. Additionally, it contributes $|X'\cup Y|\cdot (a-\Delta)-|Y|$ to increase the degree of the vertices from $\Delta$ to~$a$ (note that minus~$Y$ is necessary because these degrees are already counted for the vertices in~$Z'$). Hence, the difference in the sum of the degrees from~$\TG$ to~$\TG+\TS$ is $|Y|+(|X'\cup Y|\cdot (a-\Delta))$ which is by our assumptions on $|X'|$ and~$|Y|$ an odd number, implying a contradiction.

    It remains to consider the case where $\TG+\TS$ contains two vertex-degrees that are larger than~$\Delta$. Hence, $k=|X'\cup Y|/2$. Again we consider the difference in the sum of the vertex-degree from~$\TG$ to~$\TG+\TS$. The vertex set $X'\cup Y$ is partitioned into a set $P_1$ the has degree~$a_1$ in~$\TG+\TS$ and  $P_2$ that has degree~$a_2$. Then the difference on the sum of the degrees is
    \begin{align*}
     2|Y|+&|P_1|(a_1\cdot \Delta)+|P_2|(a_2-\Delta)-|Y|\\
     |Y|+&\underbrace{P_1}_{\text{odd}}\underbrace{(a_1- \Delta)}_{\text{odd}}+\underbrace{|P_2|}_{\text{odd}}\underbrace{(a_2-\Delta)}_{\text{even}}
    \end{align*}
    Since $|Y|$ is odd exactly one of $|P_1|(a_1\cdot \Delta)$ or $|P_2|(a_2-\Delta)$ has to be an odd number. This implies together with the assumption that $|X'\cup Y|=|P_1\cup P_2|$ is an even number that both $|P_1|$ and~$|P_2|$ are odd numbers. Hence, since $k=|X'\cup Y|/2$ this implies that~$k$ is an odd number. However, in this case we can make~$G'$ $k$-anonymous by just adding a co-matching on $k-|Y|$ vertices in~$X'$.\qed
\end{proof}}
\medskip
\emph{Remark.}
We strongly conjecture that the bound in  \autoref{lem:degreeBound} is \emph{not} tight.
The worst example we found is a graph consisting of two disjoint cliques of size~$\Delta$ and~$\Delta+1$, respectively, and setting~$k=n$.
The only \kInsertSet for this instance makes the whole graph a clique and, thus, doubles the degree.
We conjecture that the bound can be improved to~$2\Delta$.

By \autoref{lem:degreeBound} we may assume that the the maximum degree is not forced to increased by more than~$\degreeBound-\Delta$.
We now have all ingredients to prove that the local exchange heuristic always realizes ``large'' demand functions.

\begin{theorem}
Let $G=(V,E)$ be a graph with maximum degree~$\Delta$ and let \mbox{$\dem:V\rightarrow \N$} be a demand function such that $\max_{v\in V}\dem(v)+\deg_G(v)\le \degreeBound$.
The local exchange heuristic always realizes~$\dem$ if $\sum_{v\in V}\dem(v)\ge 20\Delta^4+4\Delta^2$.
\end{theorem}
{\begin{proof}
 Towards a contradiction, assume that the local exchange heuristic gets stuck at some point such that no edge is insertable and no further exchange operation can be performed.
 Denote by~$V^\dem$ the set of vertices still having a non-zero demand in~$V^\dem$.
 Let~$S$ be the set of new edges already added at this point.
 We make a case distinction on the size of~$S$ and~$V^\dem$.

 \textbf{Case~1:} $|V^\dem|\ge \degreeBound+2$\\
 In this case consider any vertex~$v\in V^\dem$ and observe that it cannot have more than~$\degreeBound$ neighbors in~$G$ and~$S$ together. Hence, there is a vertex~$u\in V^\dem$ such that $\{v,u\}$ is insertable.

 \textbf{Case~2:} $|S|\ge 8\Delta^4$\\
 Consider first the subcase where $V^\dem$ consists only of one vertex~$v$.
 Hence, the demand of~$v$ is at least two and since no exhange operation is applicable, for all edges $\{u,w\}\in S$ it holds that either $\{v,u\}$ or $\{v,w\}$ is not insertable.
 However, as~$v$ can have at most~$\degreeBound$ neighbors in~$G$ and~$S$ together, from the lower bound on~$S$ it follows that there is a least one edge~$\{u,w\}\in S$ where the exchange operation can be applied.

 In the last subcase assume that there are two vertices~$v_1,v_2\in V^\dem$.
 Again, as no exchange operation can be applied, for each edge $\{u,w\}\in S$ it holds that either one of the edges $\{v_1,u\}$ and~$\{v_2,w\}$ or one of $\{v_1,w\}$ and $\{v_2,u\}$ is not insertable.
 However, for vertex~$v_1$~($v_2$) it holds that there are less than~$(\degreeBound)^2=4\Delta^4$ edges in~$S$ which contain a neighbor of~$v_1$ ($v_2$, resp.).
 Hence, from $|S|\ge 8\Delta^4$ it follows that there is an edge~$\{u,w\}$ in~$S$ where each of $\{u,w\}$ is non-adjacent to each of~$\{v_1,v_2\}$ and thus the exchange operation can be applied. This~completes~Case~2.

 Since the demand of each vertex in~$V^\dem$ is at most~$\degreeBound$, from~Case~1 and $\xi:=\sum_{v\in V}\dem(v)\ge 20\Delta^4+4\Delta^2$ it follows that $|S|$~is at least
 \[
 \frac{\xi-|V^\dem|\cdot \degreeBound}{2}\ge \frac{\xi-(\degreeBound+2)\degreeBound}{2}\ge 8\Delta^4.
 \]
Hence, Case~2 applies and this causes a contradiction to the assumption that the local exhange heuristic got stuck at some point.\qed
\end{proof}}


\section{Appendix: Full experimental results}

\paragraph{Further Real-World Instances.}
Besides the DIMACS instances we consider coauthor networks derived from the DBLP dataset where the vertices represent authors and the edges represent co-authorship in at least one paper.
The DBLP-dataset was generated on February 2012 following the documentation from \url{http://dblp.uni-trier.de/xml/}.
As it turned out that this DBLP graph is too large for our exact approach, we derived the following subnetworks:
First, we made the graph sparser by making two vertices adjacent if the corresponding two authors are co-authors in at least two, three or more papers instead of one paper.
We denote these graphs by graph\_thres\_1 (originial graph), graph\_thres\_2, $\ldots$, graph\_thres\_5.
Second, we just considered papers that appeared in some algorithm engineering or algorithm theory conference (the exact conference list here is: SEA, WEA, ALENEX, ESA, SODA, WADS, COCOON, ISAAC, WALCOM, AAIM, FAW, SWAT) and removed all isolated vertices.
The resulting graph is denoted by graphConference.

\begin{table}[t]
\renewcommand{\tabcolsep}{10pt}
	\caption{Graph parameters of the real world networks.} \label{tbl:graphParams}
	\begin{tabular}{ll|rrr}
 & graph & $n$ & $m$ & $\Delta$ \\ \thickhline
DIMACS graps & coPapersDBLP & 540,486 & 15,245,729 & 3,299 \\
 & coPapersCiteseer & 434,102 & 16,036,720 & 1,188 \\
 & coAuthorsDBLP & 299,067 & 977,676 & 336 \\
 & citationCiteseer & 268,495 & 1,156,647 & 1,318 \\
 & coAuthorsCiteseer & 227,320 & 814,134 & 1,372 \\ \hline
DBLP subgraphs & graph\_thres\_01 & 715,633 & 2511,988 & 804 \\
 & graph\_thres\_02 & 282,831 & 640,697 & 201 \\
 & graph\_thres\_03 & 167,006 & 293,796 & 123 \\
 & graph\_thres\_04 & 112,949 & 168,524 & 88 \\
 & graph\_thres\_05 & 81,519 & 107,831 & 71 \\
 & graphConference & 5,599 & 8,492 & 53 \\ \thickhline
	\end{tabular}
\end{table}

\begin{center}\footnotesize
	\centering
	\begin{longtable}{cr|rrr|rr|rrr}
	\caption{Full list of experimental results on real-world instances with enabled data reduction. We use the following abbreviations: CH for clustering-heuristic of \citet{LSB12}, OH for our upper bound heuristic, OPT for optimal value for the \kDegAnon problem, and DP for dynamic program for the \kDSRA problem.
	If the time entry for DP is empty, then we could not solve the \kDSRA instance within one-hour and the DP bounds display the lower and upper bounds computed so far.
	If OPT is empty, then either the \kDSRA solutions could not be realized or the \kDSRA instance could not be solved within one hour.}\label{tab:fullRealWorldResults} \\
&  & \multicolumn{3}{c|}{solution size} & \multicolumn{2}{c|}{DP bounds} & \multicolumn{3}{c}{time (in seconds)} \\
graph & k & CH & OH & OPT & lower & upper & CH & OH & DP \\ \thickhline \endfirsthead
&  & \multicolumn{3}{c|}{solution size} & \multicolumn{2}{c|}{DP bounds} & \multicolumn{3}{c}{time (in seconds)} \\
graph & k & CH & OH & OPT & lower & upper & CH & OH & DP \\ \thickhline \endhead
\multirow{14}{*}{\rotatebox[origin=c]{90}{citationCiteseer}} & 2 & 690 & 458 & \emptyBox & 319 & 457 & 1.49 & 0.53 & \emptyBox \\
\nopagebreak & 3 & 1,187 & 704 & \emptyBox & 526 & 698 & 2.59 & 0.43 & \emptyBox \\
\nopagebreak & 4 & 1,689 & 1,133 & \emptyBox & 805 & 1,111 & 1.27 & 1.08 & \emptyBox \\
\nopagebreak & 5 & 2,508 & 1,807 & \emptyBox & 1,224 & 1,756 & 1.64 & 4.14 & \emptyBox \\
\nopagebreak & 7 & 3,860 & 2,698 & \emptyBox & 1,815 & 2,653 & 0.92 & 5.01 & \emptyBox \\
\nopagebreak & 10 & 6,543 & 4,966 & \emptyBox & 3,174 & 4,769 & 1.08 & 43.05 & \emptyBox \\
\nopagebreak & 15 & 10,491 & 8,249 & \emptyBox & 5,117 & 7,931 & 1.15 & 77.71 & \emptyBox \\
\nopagebreak & 20 & 15,934 & 12,778 & \emptyBox & 7,770 & 12,058 & 1.14 & 579.94 & \emptyBox \\
\nopagebreak & 30 & 24,099 & 19,800 & \emptyBox & 11,835 & 18,797 & 1.41 & 1,538.03 & \emptyBox \\
\nopagebreak & 50 & 45,257 & \emptyBox & \emptyBox & 22,316 & 35,986 & 1.4 & \emptyBox & \emptyBox \\
\nopagebreak & 100 & 98,688 & \emptyBox & \emptyBox & 49,041 & 81,438 & 2.41 & \emptyBox & \emptyBox \\
\nopagebreak & 150 & 154,753 & \emptyBox & \emptyBox & 127,994 & 127,994 & 5.25 & \emptyBox & \emptyBox \\
\nopagebreak & 200 & 211,427 & \emptyBox & \emptyBox & 174,040 & 174,040 & 6.78 & \emptyBox & \emptyBox \\ \hline
\multirow{14}{*}{\rotatebox[origin=c]{90}{coAuthorsCiteseer}} & 2 & 1,163 & 1,002 & 1,002 & 1,002 & 1,002 & 0.74 & 0.26 & 1.44 \\
\nopagebreak & 3 & 2,156 & 1,837 & 1,836 & 1,836 & 1,836 & 0.98 & 0.33 & 85.87 \\
\nopagebreak & 4 & 3,102 & 2,979 & \emptyBox & 2,977 & 2,977 & 0.98 & 0.97 & 842.96 \\
\nopagebreak & 5 & 4,413 & 4,177 & \emptyBox & 4,163 & 4,163 & 0.92 & 1.79 & 1,417.36 \\
\nopagebreak & 7 & 6,304 & 6,021 & \emptyBox & 3,701 & 5,979 & 0.62 & 3.11 & \emptyBox \\
\nopagebreak & 10 & 9,716 & 9,332 & \emptyBox & 9,283 & 9,283 & 0.95 & 7.42 & 3,465.51 \\
\nopagebreak & 15 & 15,843 & 15,117 & \emptyBox & 7,920 & 14,932 & 0.85 & 48.58 & \emptyBox \\
\nopagebreak & 20 & 21,630 & 20,745 & \emptyBox & 20,584 & 20,584 & 1.03 & 68.74 & 64.78 \\
\nopagebreak & 30 & 34,065 & 32,719 & \emptyBox & 32,424 & 32,424 & 1.34 & 242.24 & 1,330.47 \\
\nopagebreak & 50 & 59,251 & 57,002 & \emptyBox & \emptyBox & \emptyBox & 1.97 & 2,546.24 & \emptyBox \\
\nopagebreak & 100 & 122,996 & \emptyBox & \emptyBox & 61,288 & 114,512 & 3.77 & \emptyBox & \emptyBox \\
\nopagebreak & 150 & 187,977 & \emptyBox & \emptyBox & 171,411 & 171,411 & 7.95 & \emptyBox & \emptyBox \\
\nopagebreak & 200 & 252,170 & \emptyBox & \emptyBox & 125,946 & 226,200 & 10.43 & \emptyBox & \emptyBox \\ \hline
\multirow{14}{*}{\rotatebox[origin=c]{90}{coAuthorsDBLP}} & 2 & 97 & 62 & \emptyBox & 61 & 61 & 1.47 & 0.08 & 0.04 \\
\nopagebreak & 3 & 253 & 180 & 179 & 179 & 179 & 1.26 & 0.06 & 0.12 \\
\nopagebreak & 4 & 344 & 231 & \emptyBox & 230 & 230 & 0.82 & 0.14 & 0.26 \\
\nopagebreak & 5 & 531 & 321 & 317 & 317 & 317 & 1.41 & 0.29 & 26.77 \\
\nopagebreak & 7 & 817 & 542 & \emptyBox & 414 & 527 & 0.63 & 0.16 & \emptyBox \\
\nopagebreak & 10 & 1,372 & 893 & \emptyBox & 869 & 869 & 1.03 & 0.48 & 1.58 \\
\nopagebreak & 15 & 2,352 & 1,549 & \emptyBox & 1,481 & 1,481 & 0.97 & 2.62 & 4.63 \\
\nopagebreak & 20 & 3,323 & 2,188 & \emptyBox & 2,081 & 2,081 & 0.84 & 7.69 & 7.64 \\
\nopagebreak & 30 & 5,381 & 3,557 & \emptyBox & 3,391 & 3,391 & 0.75 & 19.2 & 2.40 \\
\nopagebreak & 50 & 9,661 & 6,700 & \emptyBox & 4,744 & 6,042 & 0.61 & 158.36 & \emptyBox \\
\nopagebreak & 100 & 21,267 & 15,050 & \emptyBox & 10,577 & 11,981 & 1.13 & 885.79 & \emptyBox \\
\nopagebreak & 150 & 32,932 & 24,139 & \emptyBox & 16,483 & 16,483 & 0.63 & 801.69 & 15.38 \\
\nopagebreak & 200 & 45,411 & 32,925 & 21,960 & 21,960 & 21,960 & 0.62 & 271.55 & 11.94 \\ \hline
\multirow{14}{*}{\rotatebox[origin=c]{90}{coPapersCiteseer}} & 2 & 203 & 80 & \emptyBox & 78 & 78 & 9.9 & 0.1 & 0.39 \\
\nopagebreak & 3 & 394 & 136 & 136 & 136 & 136 & 9.66 & 0.11 & 0.13 \\
\nopagebreak & 4 & 668 & 231 & 231 & 231 & 231 & 9.31 & 0.12 & 0.13 \\
\nopagebreak & 5 & 998 & 327 & 327 & 327 & 327 & 10.32 & 0.19 & 0.17 \\
\nopagebreak & 7 & 1,915 & 657 & 657 & 657 & 657 & 7 & 0.35 & 0.36 \\
\nopagebreak & 10 & 2,533 & 960 & 960 & 960 & 960 & 8.83 & 0.74 & 0.72 \\
\nopagebreak & 15 & 5,147 & 1,847 & 1,845 & 1,845 & 1,845 & 7.96 & 2.17 & 2.17 \\
\nopagebreak & 20 & 6,829 & 2,627 & 2,627 & 2,627 & 2,627 & 8.4 & 4.1 & 4.13 \\
\nopagebreak & 30 & 11,667 & 4,273 & 4,273 & 4,273 & 4,273 & 7.62 & 9.58 & 9.47 \\
\nopagebreak & 50 & 22,795 & 9,312 & 9,311 & 9,311 & 9,311 & 7.48 & 47.5 & 47.51 \\
\nopagebreak & 100 & 51,456 & 22,030 & 22,007 & 22,007 & 22,007 & 5.97 & 263.95 & 264.55 \\
\nopagebreak & 150 & 81,101 & 36,011 & 35,881 & 35,881 & 35,881 & 5.51 & 490.15 & 487.55 \\
\nopagebreak & 200 & 113,526 & 51,379 & 51,361 & 51,361 & 51,361 & 4.71 & 995.8 & 984.53 \\ \hline
\multirow{14}{*}{\rotatebox[origin=c]{90}{coPapersDBLP}} & 2 & 1,890 & 1,747 & \emptyBox & 950 & 1,733 & 11.28 & 2.13 & \emptyBox \\
\nopagebreak & 3 & 3,418 & 3,065 & \emptyBox & 1,683 & 3,030 & 11.66 & 4.75 & \emptyBox \\
\nopagebreak & 4 & 6,236 & 5,551 & \emptyBox & 2,996 & 5,497 & 10.36 & 11.63 & \emptyBox \\
\nopagebreak & 5 & 9,085 & 8,219 & \emptyBox & 4,414 & 8,121 & 10.66 & 28.83 & \emptyBox \\
\nopagebreak & 7 & 12,166 & 10,764 & \emptyBox & 5,897 & 10,615 & 8.39 & 55.55 & \emptyBox \\
\nopagebreak & 10 & 19,631 & 17,571 & \emptyBox & 9,557 & 17,328 & 9.95 & 149.56 & \emptyBox \\
\nopagebreak & 15 & 31,663 & 28,538 & \emptyBox & 15,403 & 27,991 & 9.38 & 729.36 & \emptyBox \\
\nopagebreak & 20 & 43,637 & 39,722 & \emptyBox & 21,469 & 39,048 & 10.24 & 1,372.12 & \emptyBox \\
\nopagebreak & 30 & 70,590 & \emptyBox & \emptyBox & 34,714 & 62,898 & 11.2 & \emptyBox & \emptyBox \\
\nopagebreak & 50 & 122,378 & \emptyBox & \emptyBox & 60,535 & 110,295 & 14.37 & \emptyBox & \emptyBox \\
\nopagebreak & 100 & 258,230 & \emptyBox & \emptyBox & 128,143 & 233,508 & 22.16 & \emptyBox & \emptyBox \\
\nopagebreak & 150 & 401,143 & \emptyBox & \emptyBox & 240,068 & 359,665 & 46.44 & \emptyBox & \emptyBox \\
\nopagebreak & 200 & 540,505 & \emptyBox & \emptyBox & 268,409 & 485,386 & 61.74 & \emptyBox & \emptyBox \\ \hline
\multirow{14}{*}{\rotatebox[origin=c]{90}{graph\_thres\_01}} & 2 & 290 & 179 & \emptyBox & 166 & 176 & 7.16 & 0.24 & \emptyBox \\
\nopagebreak & 3 & 748 & 483 & \emptyBox & 372 & 475 & 6.76 & 2.81 & \emptyBox \\
\nopagebreak & 4 & 1,154 & 729 & \emptyBox & 547 & 714 & 5.67 & 0.97 & \emptyBox \\
\nopagebreak & 5 & 1,621 & 1,009 & \emptyBox & 745 & 989 & 5.78 & 1.89 & \emptyBox \\
\nopagebreak & 7 & 2,516 & 1,638 & \emptyBox & 1,188 & 1,600 & 4.54 & 5.59 & \emptyBox \\
\nopagebreak & 10 & 3,969 & 2,632 & \emptyBox & 1,854 & 2,524 & 4.51 & 14.46 & \emptyBox \\
\nopagebreak & 15 & 6,461 & 4,488 & \emptyBox & 3,152 & 4,263 & 4.69 & 180.11 & \emptyBox \\
\nopagebreak & 20 & 8,761 & 6,127 & \emptyBox & 4,227 & 5,731 & 4.35 & 179.17 & \emptyBox \\
\nopagebreak & 30 & 13,643 & 9,712 & \emptyBox & 6,634 & 9,277 & 4.01 & 577.17 & \emptyBox \\
\nopagebreak & 50 & 24,471 & \emptyBox & \emptyBox & 11,958 & 16,874 & 3.39 & \emptyBox & \emptyBox \\
\nopagebreak & 100 & 53,101 & \emptyBox & \emptyBox & 37,792 & 37,792 & 3.19 & \emptyBox & \emptyBox \\
\nopagebreak & 150 & 82,684 & \emptyBox & \emptyBox & 40,995 & 59,281 & 3.03 & \emptyBox & \emptyBox \\
\nopagebreak & 200 & 115,553 & \emptyBox & \emptyBox & 79,952 & 79,952 & 3.93 & \emptyBox & \emptyBox \\ \hline
\multirow{14}{*}{\rotatebox[origin=c]{90}{graph\_thres\_02}} & 2 & 32 & 16 & 16 & 16 & 16 & 0.78 & 0.03 & 0.03 \\
\nopagebreak & 3 & 111 & 50 & 50 & 50 & 50 & 0.97 & 0.04 & 0.04 \\
\nopagebreak & 4 & 165 & 80 & 80 & 80 & 80 & 0.47 & 0.04 & 0.04 \\
\nopagebreak & 5 & 258 & 115 & 115 & 115 & 115 & 0.88 & 0.05 & 0.04 \\
\nopagebreak & 7 & 392 & 178 & 178 & 178 & 178 & 0.42 & 0.04 & 0.05 \\
\nopagebreak & 10 & 609 & 288 & 288 & 288 & 288 & 0.47 & 0.07 & 0.06 \\
\nopagebreak & 15 & 1,098 & 514 & \emptyBox & 502 & 502 & 0.45 & 0.5 & 0.13 \\
\nopagebreak & 20 & 1,564 & 784 & \emptyBox & 754 & 754 & 0.45 & 0.7 & 0.65 \\
\nopagebreak & 30 & 2,581 & 1,389 & \emptyBox & 1,299 & 1,299 & 0.41 & 3.79 & 358.02 \\
\nopagebreak & 50 & 5,094 & 2,729 & 2,410 & 2,410 & 2,410 & 0.38 & 4.13 & 0.52 \\
\nopagebreak & 100 & 11,281 & 6,764 & 5,419 & 5,419 & 5,419 & 0.34 & 6.77 & 1.57 \\
\nopagebreak & 150 & 18,179 & 11,347 & 8,709 & 8,709 & 8,709 & 0.54 & 17.61 & 3.11 \\
\nopagebreak & 200 & 24,950 & 16,156 & 11,921 & 11,921 & 11,921 & 0.83 & 35.33 & 4.04 \\ \hline
\multirow{14}{*}{\rotatebox[origin=c]{90}{graph\_thres\_03}} & 2 & 38 & 25 & 25 & 25 & 25 & 0.26 & 0.02 & 0.02 \\
\nopagebreak & 3 & 71 & 45 & \emptyBox & 44 & 44 & 0.29 & 0.26 & 0.02 \\
\nopagebreak & 4 & 121 & 69 & 68 & 68 & 68 & 0.18 & 0.04 & 0.03 \\
\nopagebreak & 5 & 166 & 94 & 94 & 94 & 94 & 0.24 & 0.03 & 0.06 \\
\nopagebreak & 7 & 276 & 154 & \emptyBox & 150 & 150 & 0.18 & 0.08 & 0.11 \\
\nopagebreak & 10 & 439 & 240 & \emptyBox & 238 & 238 & 0.17 & 0.22 & 0.20 \\
\nopagebreak & 15 & 763 & 382 & \emptyBox & 374 & 374 & 0.17 & 0.73 & 0.08 \\
\nopagebreak & 20 & 1,015 & 536 & \emptyBox & 514 & 514 & 0.16 & 0.26 & 0.61 \\
\nopagebreak & 30 & 1,668 & 830 & 797 & 797 & 797 & 0.15 & 0.45 & 0.07 \\
\nopagebreak & 50 & 3,100 & 1,543 & 1,491 & 1,491 & 1,491 & 0.15 & 0.44 & 0.14 \\
\nopagebreak & 100 & 7,150 & 3,476 & 3,380 & 3,380 & 3,380 & 0.15 & 0.82 & 0.47 \\
\nopagebreak & 150 & 11,347 & 6,359 & 5,354 & 5,354 & 5,354 & 0.23 & 4.12 & 0.67 \\
\nopagebreak & 200 & 15,349 & 8,874 & 7,450 & 7,450 & 7,450 & 0.33 & 6.57 & 1.03 \\ \hline
\multirow{14}{*}{\rotatebox[origin=c]{90}{graph\_thres\_04}} & 2 & 28 & 15 & 15 & 15 & 15 & 0.11 & 0.01 & 0.01 \\
\nopagebreak & 3 & 46 & 25 & 25 & 25 & 25 & 0.13 & 0.05 & 0.02 \\
\nopagebreak & 4 & 82 & 38 & 38 & 38 & 38 & 0.1 & 0.03 & 0.02 \\
\nopagebreak & 5 & 85 & 39 & 39 & 39 & 39 & 0.11 & 0.01 & 0.02 \\
\nopagebreak & 7 & 169 & 77 & 77 & 77 & 77 & 0.1 & 0.01 & 0.02 \\
\nopagebreak & 10 & 267 & 134 & 134 & 134 & 134 & 0.09 & 0.03 & 0.02 \\
\nopagebreak & 15 & 481 & 230 & 229 & 229 & 229 & 0.09 & 0.07 & 0.02 \\
\nopagebreak & 20 & 696 & 323 & \emptyBox & 321 & 321 & 0.08 & 0.15 & 0.02 \\
\nopagebreak & 30 & 1,145 & 548 & 548 & 548 & 548 & 0.08 & 0.03 & 0.03 \\
\nopagebreak & 50 & 2,174 & 1,041 & 1,039 & 1,039 & 1,039 & 0.08 & 0.08 & 0.07 \\
\nopagebreak & 100 & 5,177 & 2,430 & 2,398 & 2,398 & 2,398 & 0.09 & 0.35 & 0.20 \\
\nopagebreak & 150 & 8,165 & 4,071 & 3,813 & 3,813 & 3,813 & 0.16 & 1.2 & 0.30 \\
\nopagebreak & 200 & 11,280 & 6,270 & 5,327 & 5,327 & 5,327 & 0.2 & 3.06 & 0.49 \\ \hline
\multirow{14}{*}{\rotatebox[origin=c]{90}{graph\_thres\_05}} & 2 & 17 & 9 & 8 & 8 & 8 & 0.07 & 0.01 & 0.01 \\
\nopagebreak & 3 & 28 & 12 & 12 & 12 & 12 & 0.06 & 0.01 & 0.01 \\
\nopagebreak & 4 & 51 & 25 & 25 & 25 & 25 & 0.06 & 0.01 & 0.01 \\
\nopagebreak & 5 & 64 & 31 & 31 & 31 & 31 & 0.06 & 0.01 & 0.01 \\
\nopagebreak & 7 & 135 & 60 & 60 & 60 & 60 & 0.06 & 0.01 & 0.01 \\
\nopagebreak & 10 & 207 & 102 & 102 & 102 & 102 & 0.06 & 0.01 & 0.01 \\
\nopagebreak & 15 & 344 & 176 & 174 & 174 & 174 & 0.05 & 0.03 & 0.01 \\
\nopagebreak & 20 & 537 & 261 & 261 & 261 & 261 & 0.05 & 0.01 & 0.01 \\
\nopagebreak & 30 & 919 & 429 & 429 & 429 & 429 & 0.05 & 0.02 & 0.02 \\
\nopagebreak & 50 & 1,791 & 857 & 856 & 856 & 856 & 0.05 & 0.05 & 0.04 \\
\nopagebreak & 100 & 4,248 & 1,956 & 1,953 & 1,953 & 1,953 & 0.05 & 0.12 & 0.11 \\
\nopagebreak & 150 & 6,906 & 3,214 & 3,122 & 3,122 & 3,122 & 0.1 & 0.4 & 0.15 \\
\nopagebreak & 200 & 9,581 & 4,590 & 4,375 & 4,375 & 4,375 & 0.16 & 0.88 & 0.24 \\ \hline
\multirow{14}{*}{\rotatebox[origin=c]{90}{graphConference}} & 2 & 64 & 58 & 58 & 58 & 58 & 0.04 & 0.04 & 0.04 \\
\nopagebreak & 3 & 104 & 93 & 93 & 93 & 93 & 0.01 & 0 & 0.00 \\
\nopagebreak & 4 & 192 & 154 & 154 & 154 & 154 & 0.01 & 0 & 0.02 \\
\nopagebreak & 5 & 261 & 216 & \emptyBox & 214 & 214 & 0.01 & 0.02 & 0.02 \\
\nopagebreak & 7 & 417 & 346 & \emptyBox & 344 & 344 & 0.01 & 0.04 & 0.05 \\
\nopagebreak & 10 & 690 & 553 & \emptyBox & 536 & 536 & 0.06 & 0.22 & 0.78 \\
\nopagebreak & 15 & 1,122 & 880 & \emptyBox & 851 & 851 & 0.01 & 0.15 & 0.15 \\
\nopagebreak & 20 & 1,630 & 1,204 & \emptyBox & 1,157 & 1,157 & 0.01 & 0.66 & 1.01 \\
\nopagebreak & 30 & 2,474 & 1,872 & \emptyBox & 1,732 & 1,732 & 0.07 & 11.21 & 1.25 \\
\nopagebreak & 50 & 4,454 & 3,177 & \emptyBox & 2,694 & 2,694 & 0.01 & 62.66 & 0.90 \\
\nopagebreak & 100 & 8,506 & 6,428 & 4,636 & 4,636 & 4,636 & 0.04 & 2.94 & 0.27 \\
\nopagebreak & 150 & 15,699 & \emptyBox & 7,225 & 7,225 & 7,225 & 0.41 & \emptyBox & \emptyBox \\
\nopagebreak & 200 & 21,267 & \emptyBox & 9,942 & 9,942 & 9,942 & 0.75 & \emptyBox & \emptyBox \\ \hline
\multirow{14}{*}{\rotatebox[origin=c]{90}{citationCiteseer}} & 2 & 690 & 457 & \emptyBox & 311 & 457 & 1.69 & 0.47 & \emptyBox \\
\nopagebreak & 3 & 1,187 & 698 & \emptyBox & 519 & 698 & 1.59 & 0.37 & \emptyBox \\
\nopagebreak & 4 & 1,690 & 1,119 & \emptyBox & 802 & 1,111 & 1 & 0.71 & \emptyBox \\
\nopagebreak & 5 & 2,502 & 1,816 & \emptyBox & 1,219 & 1,756 & 1.73 & 4.61 & \emptyBox \\
\nopagebreak & 7 & 3,864 & 2,735 & \emptyBox & 1,814 & 2,653 & 0.9 & 7.3 & \emptyBox \\
\nopagebreak & 10 & 6,537 & 4,882 & \emptyBox & 3,173 & 4,769 & 1.04 & 29.61 & \emptyBox \\
\nopagebreak & 15 & 10,478 & 8,365 & \emptyBox & 5,116 & 7,931 & 0.95 & 101.16 & \emptyBox \\
\nopagebreak & 20 & 15,889 & 12,800 & \emptyBox & 7,769 & 12,058 & 0.83 & 579.24 & \emptyBox \\
\nopagebreak & 30 & 24,166 & 19,988 & \emptyBox & 11,834 & 18,797 & 1.28 & 1,659.57 & \emptyBox \\
\nopagebreak & 50 & 45,214 & \emptyBox & \emptyBox & 22,316 & 35,986 & 1.25 & \emptyBox & \emptyBox \\
\nopagebreak & 100 & 98,683 & \emptyBox & \emptyBox & 49,040 & 81,438 & 2.23 & \emptyBox & \emptyBox \\
\nopagebreak & 150 & 154,753 & \emptyBox & \emptyBox & 77,049 & 127,994 & 5.31 & \emptyBox & \emptyBox \\
\nopagebreak & 200 & 211,215 & \emptyBox & \emptyBox & 105,102 & 174,040 & 6.91 & \emptyBox & \emptyBox \\ \hline
\multirow{14}{*}{\rotatebox[origin=c]{90}{coAuthorsCiteseer}} & 2 & 1,163 & 1,002 & 1,002 & 1,002 & 1,002 & 0.65 & 0.26 & 45.90 \\
\nopagebreak & 3 & 2,156 & 1,837 & 1,836 & 1,836 & 1,836 & 0.73 & 0.31 & 270.99 \\
\nopagebreak & 4 & 3,103 & 2,983 & \emptyBox & 2,977 & 2,977 & 0.54 & 0.76 & 829.92 \\
\nopagebreak & 5 & 4,416 & 4,173 & \emptyBox & 4,163 & 4,163 & 0.59 & 1.67 & 2,491.17 \\
\nopagebreak & 7 & 6,298 & 6,008 & \emptyBox & 3,168 & 5,979 & 0.62 & 2.51 & \emptyBox \\
\nopagebreak & 10 & 9,716 & 9,331 & \emptyBox & 4,902 & 9,283 & 0.69 & 7.98 & \emptyBox \\
\nopagebreak & 15 & 15,858 & 15,074 & \emptyBox & 7,910 & 14,932 & 0.78 & 41.94 & \emptyBox \\
\nopagebreak & 20 & 21,633 & 20,781 & \emptyBox & 10,823 & 20,584 & 0.89 & 87.45 & \emptyBox \\
\nopagebreak & 30 & 34,094 & 32,746 & \emptyBox & 17,041 & 32,424 & 1.31 & 268.99 & \emptyBox \\
\nopagebreak & 50 & 59,244 & 56,849 & \emptyBox & \emptyBox & \emptyBox & 1.88 & 2,541.73 & \emptyBox \\
\nopagebreak & 100 & 122,996 & \emptyBox & \emptyBox & 61,280 & 114,512 & 3.34 & \emptyBox & \emptyBox \\
\nopagebreak & 150 & 187,829 & \emptyBox & \emptyBox & 93,528 & 171,411 & 7.92 & \emptyBox & \emptyBox \\
\nopagebreak & 200 & 252,011 & \emptyBox & \emptyBox & 125,903 & 226,200 & 10.38 & \emptyBox & \emptyBox \\
 \thickhline
	\end{longtable}
\end{center}

\end{document}